\documentclass[12pt]{iopart}

\usepackage{setstack,amsthm,amsfonts,amssymb,times,bbm,bm,graphicx,dsfont,bbold}
\usepackage{etoolbox}
\usepackage{xcolor}
\usepackage{hyperref}
\usepackage{cite}
\usepackage{array}
\newcolumntype{L}{>{$}l<{$}} % text mode "l" in an "array"
\usepackage{braket}
\usepackage{latexsym}
\usepackage[T1]{fontenc}
\usepackage[utf8x]{inputenc}

\newtheorem{theorem}{Theorem}

\newtheorem{lemma}[theorem]{Lemma}
\newtheorem{corollary}[theorem]{Corollary}

\newtheorem{definition}[theorem]{Definition}

\newtheorem{remark}[theorem]{Remark}
\newtheorem{assumption}[theorem]{Assumption}

\newcommand{\eqref}[1]{(\ref{#1})}
\newcommand{\bd}[1]{\boldsymbol{ #1 }}

\renewcommand\bra[1]{{\langle{#1}|}}
\renewcommand\ket[1]{{|{#1}\rangle}}

\usepackage{subfig}
\def\RR{\mathbbm{R}}
\def\CC{\mathbbm{C}}

\def\mk {\mathfrak}

\def\Ho{\mathcal{H}_1}
\def\HN{\wedge^N[\Ho]}

\def\mH {\mathcal{H}}
\def\kpsi {\ket{\Psi}}
\def\im {{\rm{Im}}}

\begin{document}

\title[]{Implications of pinned occupation numbers for natural orbital expansions. II: Rigorous derivation and extension to non-fermionic systems}

\author{Tomasz Maci\k{a}\.zek}
\ead{maciazek@cft.edu.pl}
\address{Center for Theoretical Physics, Polish Academy of Sciences, Al.~Lotnik\'ow 32/46, 02-668 Warszawa, Poland}
\address{School of Mathematics, University of Bristol, Bristol BS8 1TW, UK}

\author{Adam Sawicki}
\ead{a.sawicki@cft.edu.pl}
\address{Center for Theoretical Physics, Polish Academy of Sciences, Al.~Lotnik\'ow 32/46, 02-668 Warszawa, Poland}

\author{David Gross}
%\ead{david.gross@thp.uni-koeln.de}
\address{Institute for Theoretical Physics, University of Cologne, Germany}

\author{Alexandre Lopes}
%\ead{alexandre.lopes@zeiss.com}
\address{Carl Zeiss SMT GmbH, Rudolf-Eber-Straße 2, 73447 Oberkochen, Germany}

\author{Christian Schilling}
%\ead{c.schilling@lmu.de}
\address{Department of Physics, Arnold Sommerfeld Center for Theoretical Physics, Ludwig-Maximilians-Universit\"at M\"unchen, Theresienstrasse 37, 80333 M\"unchen, Germany}
\address{Clarendon Laboratory, University of Oxford, Parks Road, Oxford OX1 3PU, United Kingdom}
%\address{Wolfson College, University of Oxford, Linton Rd, Oxford OX2 6UD, United Kingdom}

\date{\today}

\begin{abstract}
We have explained and comprehensively illustrated in Part I that the generalized Pauli constraints suggest a natural extension of the concept of active spaces. In the present Part II, we provide rigorous derivations of the theorems involved therein. This will offer in particular deeper insights into the underlying mathematical structure and will explain why the saturation of generalized Pauli constraints implies a specific simplified structure of the corresponding many-fermion quantum state. Moreover, we extend the results of Part I to non-fermionic multipartite quantum systems, revealing that extremal single-body information has always strong implications for the multipartite quantum state. In that sense, our work also confirms that pinned quantum systems define new physical entities and the presence of pinnings reflect the existence of (possibly hidden) ground state symmetries.
\end{abstract}

%\pacs{}

\maketitle
\section{Introduction and brief recap of the notation}\label{sec:intro}
We consider the $N$-fermion Hilbert space $\HN$, where $\Ho$ is the underlying $d$-dimensional one-particle Hilbert space. If not stated otherwise, all states in this paper are not necessarily normalised. To each quantum state $\ket{\Psi} \in \HN$ we can assign its one particle reduced density operator $\rho_1$ which is obtained by tracing out all except one fermion,
\begin{equation}\label{1RDO}
\rho_1 \equiv N\,\mbox{Tr}_{N-1}[|\Psi\rangle\!\langle\Psi|] \equiv \sum_{j=1}^d\,n_j\,|j\rangle\! \langle j|\,.
\end{equation}
Equation \eqref{1RDO} gives rise to the \emph{natural occupation numbers} (NONs) $n_j$ and the \emph{natural orbitals} (NOs) $|j\rangle$, the corresponding eigenstates \cite{Loewdin55,DavidsonRev}. This terminology also motivates the normalization $\mbox{Tr}_1[\rho_1]=n_1+\ldots+n_d =N$ which allows us to interpret the eigenvalues of $\rho_1$ as occupation numbers, the occupancies of the natural orbitals. The NOs form an  orthonormal basis $\mathcal{B}_1$ for $\Ho$ which is unique as long as the NONs are non-degenerate.

Including the physically relevant case of degenerate NONs, any such NO basis $\mathcal{B}_1$ induces an orthonormal basis for  $\HN$ given by the family of ${d}\choose{N}$ configuration states $\ket{i_1,\ldots,i_N}\equiv f_{i_1}^\dagger \ldots f_{i_N}^\dagger\ket{0}$, with $1 \leq i_1 < i_2 < \ldots < i_N \leq d$. Here, $\ket{0}$ denotes the vacuum state and $f_j^\dagger$ is the creation operator of a fermion in the NO $\ket{j}$.
Since $\mathcal{B}_N$ is a basis for $\HN$ we can expand every quantum state in $\HN$ uniquely with respect to $\mathcal{B}_N$, in particular also $|\Psi\rangle$ ($\bd{i}\equiv (i_1,\ldots, i_N)$)
\begin{equation}\label{PsiNO}
\ket{\Psi} = \sum_{\bd{i}}\, c_{\bd{i}}\,\ket{\bd{i}}\,.
\end{equation}
The expansion \eqref{PsiNO} is self-consistent in the sense that the coefficients $c_{\bd{i}}$ are such that the corresponding one-particle reduced density operator is diagonal in its own natural orbital basis. 
Actually, in the natural expansion (\ref{PsiNO}) some of the coefficients might be zero. In the following, we will often distinguish this set from those configurations $\bd{i}$ which contribute to the expansion of $\kpsi$ also called $\ket{\Psi}$'s {\it natural support}, $\mathrm{Supp}_{\mathcal{B}_1}(\kpsi)$, of $\kpsi$
\begin{equation}\label{def:support}
\mathrm{Supp}_{\mathcal{B}_1}(\kpsi):=\{\bd{i}:\ \ket{\bd i}\in\mathcal{B}_N{\rm\ and}\ \braket{\Psi|\bd i}\neq0\}.
\end{equation}
Clearly, in case of degenerate NONs, the support of $\kpsi$ may depend on the specific choice $\mathcal{B}_1$ of natural orbitals.

One particular instance of a reduction of natural support is based on the presence of pinning. To recall those main findings of Part I \cite{Pin1}, let us first recall that the set of one-particle density matrices $\rho_1$ corresponding to some $\ket{\Psi} \in \HN$ is described by the generalized Pauli constraints, i.e.~a finite set of affine conditions,
\begin{equation}\label{GPC}
  D_i(\bd{n}) \equiv \kappa_i^{(0)} +  \bd{\kappa}_i \cdot \bd{n} \equiv \kappa_i^{(0)} + \sum_{j=1}^d \kappa_i^{(j)} n_j\geq 0\,,\,\,\,\,i=1,2,\ldots, r_{N,d}<\infty,
\end{equation}
on the vector $\bd{n} \equiv (n_j)_{j=1}^d$ of decreasingly ordered NONs. The crucial result (as illustrated in Part I \cite{Pin1}) which we will rigorously derive in the following is that the saturation of a GPC $D \geq 0$ implies structural simplifications on the corresponding $\ket{\Psi}$. As described by Theorem 6 in Part I \cite{Pin1}, one has
\begin{equation}\label{thm6}
D(\bd{n}) = 0 \quad \Rightarrow \quad \hat{D}_{\mathcal{B}_1}\ket{\Psi}=0\,,
\end{equation}
where $\hat{D}_{\mathcal{B}_1}\equiv D(\hat{n}_1,\ldots, \hat{n}_d)$ is the NO induced operator of the GPC $D$. In the case of degenerate NONs one expects (see Conjecture 9, Theorem 10 and Corollary 11 in Part I \cite{Pin1}) this to be true with respect to at least one specific choice $\mathcal{B}_1$ of NOs.
The structural implications of pinning are particularly well-pronounced in the NO expansion \eqref{PsiNO}: Since the configuration states $\ket{\bd{i}}$ are the eigenstates of $\hat{D}_{\mathcal{B}_1}$, eq.~\eqref{thm6} implies a selection rule on the contributing configurations
\begin{equation}
\forall \bd{i}: \quad D(\bd{n}_{\bd{i}}) = 0 \quad \Rightarrow \quad \bd{i} \in \mathrm{Supp}_{\mathcal{B}_1}(\kpsi)\,,
\end{equation}
where $\bd{n}_{\bd{i}}$ is $\ket{\bd{i}}$'s vector of \emph{unordered} occupation numbers,
\begin{equation}\label{NONconf}
  \bd{n}_{\bd{i}}\equiv \mbox{spec}\big(N \mbox{Tr}_{N-1}[\ket{\bd{i}}\!\bra{\bd{i}}]\big)\,,\quad \mbox{i.e.},\,\left(\bd{n}_{\bd{i}}\right)_j=
 \Bigg\{ \begin{array}{@{\kern2.5pt}lL}
    \hfill 1 & if $j \in \bd{i}$\\
          0 & if $j \not \in \bd{i}$
\end{array}\,.
\end{equation}

\section{Proofs of the main results}\label{sec:tang}
In the following we formalize our approach to deriving the consequences of pinning by generalized Pauli constraints, i.e.~to proving our main results (presented as Theorems 6,10,12 in Part I \cite{Pin1}). In particular, this will allow us to treat all possible scenarios in a systematic way and generalise our findings to multiparticle systems that also consist of particles different than fermions.

Several of our key results will rely on a close investigation of local symmetries of quantum states. The word ``local'' in this context does not refer to spatial locality but rather to a form of locality in underlying mathematical space  $\HN$: In first quantisation which is based on the embedding $\HN \leq \Ho^{\otimes^N}$, an operator $U$ is called local if it can be expressed as $U \equiv u^{\otimes^N}$, where $u$ acts on the underlying one-particle Hilbert space $\Ho$. In second quantization using fermionic creation ($f_i^\dagger$) and annihilation ($f_j$) operators referring to some orthonormal reference basis for $\Ho$, we can express such a local and unitary operator as
\begin{equation}\label{local-action}
U=e^{i\sum_{j,l=1}^d H_{jl} f_j^\dagger f_l}\,,\quad \mbox{with}\quad H_{jl}^\ast =H_{lj}\,.
\end{equation}
A local unitary operator $U$ represents by definition a local symmetry of $\ket{\Psi}$ if
\begin{equation}
U\ket{\Psi}=e^{i\phi}\ket{\Psi}{\rm\ \quad for\ some\ }\phi\in[0,2\pi).
\end{equation}
The group $S_{\ket{\Psi}}$ of local symmetries of $\ket{\Psi}$ can be identified as a subgroup of the group of unitary operators on $\Ho$,
\begin{equation}\label{def:local-symmetry}
S_{\ket{\Psi}}:=\left\{u: \Ho\rightarrow \Ho\,\mbox{unitary}\,\Big|\,\ u^{\otimes^N}\ket{\Psi}=e^{i\phi}\ket{\Psi}{\rm\ for\ some\ }\phi\in[0,2\pi) \right\}\,.
\end{equation}
%Equation (\ref{local-action}) immediately gives us the action of the generators of the local unitary operators. Assume that $U(t)=e^{t X}$ for some $X\in\mk{u}(d)$. Then $X$ is represented as operator
%\begin{eqnarray}\label{algebra-rep}
%X=\sum_{i,j=1}^d X_{ij}f_i^\dagger f_j,\ X^\dagger=-X,
%\end{eqnarray}
%where $X_{ij}=\frac{d}{dt}\big{|}_{t=0}U_{ij}(t)$ are entries of a $d\times d$ anti-hermitian matrix.
In analogy to the local symmetry group of $\ket{\Psi}$, we introduce for a one-particle reduced density operators $\rho_1$ its symmetry group as
\begin{equation}\label{def:rho-symmetry}
S_{\rho_1}:=\left\{u: \Ho\rightarrow \Ho\,\mbox{unitary}\,|\ u^\dagger \rho_1 u=\rho_1\right\}\,.
\end{equation}
Since any local transformations $\ket{\Psi}\mapsto u^{\otimes^N}\ket{\Psi}$ act by conjugation on $\rho_1$, $\rho_1\mapsto u^\dagger \rho_1 u$, any local symmetry of $\ket{\Psi}$ represents also a symmetry of $\rho_1$. In other words, we find the important inclusion relation
\begin{equation}\label{symmetries-rel}
S_{\ket{\Psi}}\subset S_{\rho_1}\,.
\end{equation}
Both groups $S_{\ket{\Psi}}$ and $S_{\rho_1}$ are actually Lie groups. Their corresponding Lie algebras $\mk{s}_{\rho_1}$ and $\mk{s}_{\ket{\Psi}}$, respectively, arise as the tangent spaces at the ``point'' $\mathds{1} \in S_{\ket{\Psi}} \subset S_{\rho_1}$ and will play a crucial role in our work.
%To be more specific, the generators of the Lie algebra $\mk{s}_{\rho_1}$ of $S_{\rho_1}$ follow as the
$\mk{s}_{\rho_1}$ is given as the algebra of all anti-hermitian operators $i h$ on the one-particle Hilbert space $\Ho$ which commute with $\rho_1$,
\begin{equation}\label{LieAlg1}
\mk{s}_{\rho_1}=\{i h\in \mk{u}(d):\ [\rho_1,h]=0\}\,.
\end{equation}
The Lie algebra $\mk{s}_{\ket{\Psi}}$ of $S_{\ket{\Psi}}$ forms then a subalgebra of $\mk{s}_{\rho_1}$, given by
\begin{equation}\label{LieAlgN}
\mk{s}_{\ket{\Psi}}=\{i h\in \mk{u}(d):\ \sum_{j,l=1}^d h_{jl}f_j^\dagger f_l \ket{\Psi}= \lambda \ket{\Psi},\,\mbox{for some}\,\lambda \in \RR\} \,\subset\, \mk{s}_{\rho_1}\,.
\end{equation}
Here, $\mk{u}(d)$ denotes the Lie algebra of the Lie group of unitary operators on $\Ho$, i.e.~$\mk{u}(d)$ is the algebra of anti-hermitian operators on $\Ho$ and $h_{jl}\equiv\bra{j}h\ket{l}$. To verify \eqref{LieAlg1}, recall that any unitary operator $u$ on $\Ho$ can be expressed as $u= e^{i h}$ for some hermitian operator $h$ and that the generators of a Lie group follow as the derivatives $\frac{\rm d u}{{\rm d} t}(t=0)$ of any one-parametric curve $u(t)\equiv e^{i t h}$.

Furthermore, let us denote by $\mu$ the map which assigns to a state $\ket{\Psi}\in \HN$ its one-particle reduced density operator
\begin{equation}\label{mu}
\mu:\ \ket{\Psi}\mapsto \rho_1.
\end{equation}
Formally, $\mu$ can be viewed as a map acting between two vector spaces where the target space is the space of hermitian $d\times d$ matrices. We identify the target space with the Lie algebra of the group $U(d)$. In other words, $\mu:\ \HN\to \iota\mk{u}(d)$. Note that in fact the image of $\mu$ is not all of $\iota \mk{u}(d)$, but it consists of positive-semidefinite matrices whose trace is equal to $N$. In this section, we will focus on regularity properties of map $\mu$. Let us next explain what we mean by regularity. Consider a one-parameter family of states (a curve) $\ket{\Psi(t)},\ t\in [-\frac{1}{2},\frac{1}{2}]$. This family gives rise to a one-parameter family of one-particle reduced density operators given by $\mu\left(\ket{\Psi(t)}\right)$. Consider next the velocity vector associated with curve $\ket{\Psi(t)}$ given by the derivative
\begin{equation}
\ket{\dot\Psi(0)}:=\frac{d}{dt}\Big{|}_{t=0}\ket{\Psi(t)}.
\end{equation}
By considering all possible curves that go through a common point $\ket{\Psi_0}$ and their velocities at $t=0$, we obtain a vector space which is the same as $\HN$. Consequently, we consider the time-derivative of the corresponding curve in the space of one-particle reduced density operators, i.e.
\begin{equation}
\dot\rho_1(0):=\frac{d}{dt}\Big{|}_{t=0}\mu\left(\ket{\Psi(t)}\right).
\end{equation}
By the chain rule, the result $\dot\rho_1(0)$ depends linearly on $\ket{\dot\Psi(0)}$ and the matrix that transforms one vector to another depends only on $\ket{\Psi_0}$ and is called the derivative matrix, $d\mu_{\ket{\Psi_0}}$. In other words,
\begin{equation}
\dot\rho_1(0):=d\mu_{\ket{\Psi_0}}\ket{\dot\Psi(0)}.
\end{equation}
The rank of the linear operator $d\mu_{\ket{\Psi_0}}$ tells us how many directions we can cover within $\iota \mk{u}(d)$ by taking all curves that go through $\ket{\Psi_0}$. Intuitively, if we are in a generic situation where the NONs of $\ket{\Psi_0}$ do not saturate any of the GPCs, the rank of $d\mu_{\ket{\Psi_0}}$ is maximal and thus equal to $d^2-1$. However, to properly justify this assertion, one has to invoke the {\it principal orbit type theorem}, a fact which is covered by Theorem \ref{stratification-properties}.

The importance of the operator $d\mu$ becomes evident when one considers pinned states, i.e.\ states whose NONs saturate at least one of the GPCs. From the sole fact that GPCs cannot be broken by a pure state, we have that the operator $d\mu$ cannot be of maximal rank when evaluated at such pinned states. In particular, the derivative vector $\dot\rho_1(0)$ cannot point out of the region of admissible one-particle reduced density operators given by the GPCs. Importantly, this phenomenon imposes tremendous restrictions on the structure of pinned quantum states as already discussed and illustrated in Part I \cite{Pin1}. The following lemma will be the point of departure for the results presented in this section. This is a variant of a well-known result which is valid in a much more general setting \cite{Guillemin} (see also \cite{SWK13, MOS13} where it was first used to study the structure of qubit states that saturate GPCs). Nevertheless, we reprove it here using more elementary arguments.

\begin{lemma}\label{thm:image}
Let $\mu:\ \HN\to i\mk{u}(d)$ be the map that assigns to a pure quantum state its one-particle reduced density operator and let $d\mu_{\ket{\Psi}}:\ \HN\to i\mk{u}(d)$ be the derivative of $\mu$.

\noindent We have
\begin{equation}\label{eq:image}
\im\ d\mu_{\ket{\Psi}}=i\left({\mk s}_{\ket{\Psi}}\right)^\perp,
\end{equation}
where
\begin{equation}\left({\mk s}_{\ket{\Psi}}\right)^\perp=\left\{ih\in \mk{u}(d):\ \tr(hh')=0{\rm\ for\ all\ }ih'\in {\mk s}_{\ket{\Psi}}\right\}.
\end{equation}
\end{lemma}
\begin{proof}
Recall the definition of the derivative. Any element of the domain of $d\mu$ can be represented as a differentiable curve $\ket{\Psi(t)}$ such that $\ket{\Psi(0)}=\kpsi$. For such a curve, $\ket{\dot\Psi(0)}\in \HN$. The map $d\mu_{\ket{\Psi}}$ acts on $\ket{\dot\Psi(0)}$ in the following way:
\begin{equation}\label{dmu-definiton}
d\mu_{\ket{\Psi}}\left(\ket{\dot\Psi(0)}\right)=\frac{d}{dt}\bigg{|}_{t=0}\mu\left(\ket{\Psi(t)}\right).
\end{equation}
Note first, that because $\mu$ does not change along the complex line through $\kpsi$, the image of $d\mu_{\kpsi}$ is invariant under complex scaling, i.e
\[\im\ d\mu_{\ket{\Psi}}=\im\ d\mu_{c\ket{\Psi}}\ {\rm for\ all}\ c\in\CC-\{0\}.\]
  Moreover, vectors proportional to $\kpsi$ belong to the kernel of $d\mu_{\kpsi}$. Hence, in order to find $\im\ d\mu_{\ket{\Psi}}$, it is enough to consider only tangent vectors corresponding to curves of the form $e^{tA}\kpsi$ for $A\in \mk{u}(\HN)$ (note, that this is the global unitary algebra).

 Our goal is to prove equation (\ref{eq:image}), which is equivalent to the fact that $\tr\left(d\mu_{\ket{\Psi}}\left(A\kpsi\right){X}\right)=0$ for all $A\in\mk{u}(\mH)$ if and only if $X\in{\mk s}_{\ket{\Psi}}$, i.e. $\left[X,\ket{\Psi}\bra{\Psi}\right]=0$. We first use the definition of $d\mu$ (formula (\ref{dmu-definiton}))
 \begin{equation}
 \fl \tr\left(d\mu_{\ket{\Psi}}\left(\ket{\dot\Psi(0)}\right){\iota X}\right)=\frac{d}{dt}\bigg{|}_{t=0} \tr\left(\mu\left(\ket{\Psi(t)}\right){\iota X}\right)=\frac{d}{dt}\bigg{|}_{t=0} \tr\left(\rho_1(\ket{\Psi(t)}) \iota X\right).
 \end{equation}
 By choosing $\ket{\Psi(t)}=e^{tA}\kpsi$, we get, that
  \begin{equation}
 \fl \tr\left(d\mu_{\ket{\Psi}}\left(A\kpsi\right){\iota X}\right)=\frac{d}{dt}\bigg{|}_{t=0} \tr\left(\rho_1\left(e^{At}\ket{\Psi}\right) \iota X\right)=\frac{d}{dt}\bigg{|}_{t=0} \tr\left(e^{-At}\left(\ket{\Psi}\bra{\Psi}\right)e^{At} \iota X\right).
 \end{equation}
In the last step, we used the fact that $\tr(\rho_1(\kpsi)X)=\tr(\ket{\Psi}\bra{\Psi}X)$ for any $X$ that acts locally. By computing the derivative and doing a cyclic permutation of matrices under the trace, we finally obtain
  \begin{equation}\label{eq:dmu}
   \tr\left(d\mu_{\ket{\Psi}}\left(A\kpsi\right){X}\right)=\tr\left(\ket{\Psi}\bra{\Psi} [A,X]\right),\ X\in \iota \mk{u}(d),\ A\in\mk{u}(\mH).
  \end{equation}

Let us first show that $\im\ d\mu_{\ket{\Psi}}\subset \left({\mk s}_{\ket{\Psi}}\right)^\perp$. To this end, assume that $X\in{\mk s}_{\ket{\Psi}}$, i.e. $\left[X,\ket{\Psi}\bra{\Psi}\right]=0$. Then for any $A\in\mk{u}(\mH)$,
\begin{equation}
\tr\left(\ket{\Psi}\bra{\Psi} [A,X]\right)=\tr\left(\ket{\Psi}\bra{\Psi}AX\right)-\tr\left(\ket{\Psi}\bra{\Psi}XA\right)=0.
\end{equation}
By doing a cyclic permutation under the first trace and commuting $X$ with $\ket{\Psi}\bra{\Psi}$ under the second trace, we get that the above expression vanishes.

Finally, we show that $\im\ d\mu_{\ket{\Psi}}\supset \left({\mk s}_{\ket{\Psi}}\right)^\perp$. To this end, we assume that there exists $X\in\iota\mk{u}(d)$ such, that for all $A\in\mk{u}(\mH)$, $\tr\left(\ket{\Psi}\bra{\Psi} [A,X]\right)=0$. By doing a cyclic permutation of matrices under the trace, we get
\begin{equation}\label{eq:proof2}
\tr\left(\ket{\Psi}\bra{\Psi} [A,X]\right)=\tr\left([X,\ket{\Psi}\bra{\Psi}] A\right).
 \end{equation}
The above trace is a non-degenerate scalar product on $\mk{u}(\mH)$. Hence by considering $\iota[X,\ket{\Psi}\bra{\Psi}]$ as an element of $\mk{u}(\mH)$, we get, that expression (\ref{eq:proof2}) vanishes for all $A\in\mk{u}(\mH)$ if and only if $[X,\ket{\Psi}\bra{\Psi}]=0$.
\end{proof}

\subsection{Non-degenerate occupation numbers}\label{subsec:tang}
Let us first cover the simpler case where the NONs are assumed to be non-degenerate, i.e. $n_1>n_2>\dots>n_d$. This case is straightforward to analyse due to the fact that a state with non-degenerate NONs has a unique basis of natural orbitals. Moreover, all local symmetry operators, $S_{\kpsi}$, are diagonal in the basis of NOs of $\kpsi$. To see this, note first that if NONs are non-degenerate, then the symmetry group of the corresponding diagonal one-particle reduced density matrices consists only of diagonal matrices, i.e.
\begin{equation}
S_{\bd{n}}=\left\{e^{i\sum_{j=1}^d\phi_j}\ket{j}\!\bra{j}:\ \phi_j\in[0,2\pi]\right\}.
\end{equation}
Finally, recall equation (\ref{symmetries-rel}) which asserts that $S_{\kpsi}$ is necessarily contained in $S_{\bd{n}}$. Recall that if a hermitian matrix $h=\sum_{k,l=1}^d H_{k,l}\ket{k}\!\bra{l}$ is a generator of local symmetry of $\kpsi$ then
\begin{equation}
\left(\sum_{k,l=1}^d H_{k,l}f^\dagger_k f_l\right)\kpsi=\lambda\kpsi\, {\mathrm{for\ some\ }}\lambda\in\RR.
\end{equation}
Moreover any generator of a symmetry of $\kpsi$ which is diagonal in NO-basis $\mathcal{B}_1$ can be written as $\hat L_{\mathcal{B}_1}$ for some linear functional $L=\sum_{i=1}^dl_in_i$. Such an operator $\hat L_{\mathcal{B}_1}$ acts on $\kpsi$ expanded in its NO-basis in the following way, as we have already stated in eq.~(13) in Part I \cite{Pin1}.
\begin{equation}\label{eq:diag-generator}
\hat L_{\mathcal{B}_1}\kpsi=\sum_{\bd{i}}c_{\bd{i}}\hat L_\psi\ket{\bd{i}}=\sum_{\bd{i}}c_{\bd{i}}(\bd{l}\cdotp\bd{n}_{\bd{i}})\ket{\bd{i}}.
\end{equation}
There, we introduced (recall also Section 2.3 of Part I \cite{Pin1}) for each configuration state $\ket{\bd{i}}$ the respective vector $\bd{n}_{\bd{i}}$ of \emph{unordered} occupation numbers,
\begin{equation}\label{NONconf}
  \bd{n}_{\bd{i}}\equiv \mbox{spec}\big(N \mbox{Tr}_{N-1}[\ket{\bd{i}}\!\bra{\bd{i}}]\big)\,,\quad \mbox{i.e.},\,\left(\bd{n}_{\bd{i}}\right)_j=
 \Bigg\{ \begin{array}{@{\kern2.5pt}lL}
    \hfill 1 & if $j \in \bd{i}$\\
          0 & if $j \not \in \bd{i}$
\end{array}\,,
\end{equation}
Hence, we obtain that $\hat L_{\mathcal{B}_1}=\sum_{i=1}^d l_i\hat n_i$ is a generator of local symmetries of $\kpsi$, i.e. $\hat L_{\mathcal{B}_1}\in {\mk s}_{\ket{\Psi}}$ if and only if there exists $\lambda\in\RR$ such that for all $\bd{i}$ that belong to $\mathrm{Supp}_{\mathcal{B}_1}(\kpsi)$ we have
\begin{equation}\label{stabiliser-cond1}
\bd{l}\cdotp\bd{n}_{\bd{i}}=\lambda.
\end{equation}
In order to find $\lambda$, we multiply both sides of (\ref{stabiliser-cond1}) by $|c_{\bd i}|^2$ and take the sum over $\bd{i}$. Using eq.~(12) from Part I \cite{Pin1} and the fact that $|c_{\bd{i}}|^2$ sum up to one, we obtain that $\lambda=\bd{l}\cdotp\bd{n}$. Hence, the condition (\ref{stabiliser-cond1}) can be conveniently rephrased as
\begin{equation}\label{stabiliser-cond2}
\bd{l}\cdotp(\bd{n}_{\bd{i}}-\bd{n})=0 {\mathrm{\ for\ all\ }}\bd{i}\in\mathrm{Supp}_{\mathcal{B}_1}(\kpsi).
\end{equation}
Summing up,
\begin{equation}\label{nondegen-stabiliser}
{\mk s}_{\ket{\Psi}}=\left\{i\sum_{j=1}^d l_i\ket{j}\!\bra{j}:\ \bd{l}\cdotp(\bd{n}_{\bd{i}}-\bd{n})=0{\rm\ for\ all\ }{\bd{i}}\in\mathrm{Supp}_{\mathcal{B}_1}(\kpsi)\right\}.
\end{equation}
From now on, we will use the shorthand notation
\begin{equation}\label{stab-notation}
\hat l_{\mathcal{B}_1}:=i\sum_{j=1}^d l_j\ket{j}\!\bra{j}.
\end{equation}
So far, we have not assumed that the NON vector is pinned to a GPC. The above results apply for any state whose NONs are non-degenerate. In particular, the support of a generic state consists of all configurations from $\mathcal{B}_N$, hence the conditions (\ref{stabiliser-cond2}) leave very little freedom for choosing the vector $\bd{l}$. In fact, generically there is only one solution, namely $\bd{l}=(1,1,\dots,1)$ which corresponds to the total particle number operator. However, as we shall see in the remaining part of this subsection, if NONs saturate a GPC, the corresponding state is necessarily not generic and has more symmetries.

Theorem \ref{thm:image} connects the above considerations with the selection rule. On the one hand, if $\kpsi$ saturates a GPC, i.e. $D_k(\bd n)=\kappa^{(0)}_k+\bd{\kappa}_k\cdotp \bd{n}=0$, then $\im\ d\mu_{\ket{\Psi}}$ cannot contain directions that have a component perpendicular to the corresponding face of polytope $\mathcal{P}$ (in the sense of the trace product). In other words, we have that
\begin{equation}\label{image-condition}
{\rm If\ } h\in \im\ d\mu_{\ket{\Psi}},{\rm\ then}\ \tr\left(h\left(\sum_{j=1}^d\kappa_k ^{(j)}\ket{j}\!\bra{j}\right)\right)=0.
\end{equation}
On the other hand, theorem \ref{thm:image} and formula (\ref{nondegen-stabiliser}) for the generators of local symmetries of $\kpsi$ tell us that $\im\ d\mu_{\ket{\Psi}}$ is an orthogonal sum of two spaces. One is the space of all hermitian purely off-diagonal matrices which we denote by $\mk{d}^\perp$
\begin{equation}
\mk{d}^\perp:=\{h:\ h^\dagger=h{\rm\ and\ }h_{ii}=0{\rm\ for\ }i=1,\dots,d\}.
\end{equation}
The other space is the space of diagonal matrices which are perpendicular to $i{\mk s}_{\ket{\Psi}}$. In other words,
\begin{equation}\label{image-decomp}
\im\ d\mu_{\ket{\Psi}}=\mk{d}^\perp\oplus{\rm Span}\left\{\sum_{j=1}^d \left((\bd{n}_{\bd{i}})_j-n_j\right)\ket{j}\!\bra{j}:\ {\bd{i}}\in\mathrm{Supp}_{\mathcal{B}_1}(\kpsi)\right\}.
\end{equation}
By comparing formula (\ref{image-decomp}) and (\ref{image-condition}), we obtain that if $\bd{i}\in\mathrm{Supp}_{\mathcal{B}_1}(\kpsi)$, then
\begin{equation}
\tr\left(\left(\sum_{j=1}^d \left((\bd{n}_{\bd{i}})_j-n_j\right)\ket{j}\!\bra{j}\right)\left(\sum_{j'=1}^d\kappa_k ^{(j')}\ket{j'}\!\bra{j'}\right)\right)=0
\end{equation}
Expanding the above formula, we have the following geometric condition for vectors $\bd{n}_{\bd{i}}$ and $\bd{\kappa}_k$
\begin{equation}
(\bd{n}_{\bd{i}}-\bd{n})\cdotp \bd{\kappa}_k=0.
\end{equation}
Noting that $\bd{n}\cdotp \bd{\kappa}_k=-\kappa_k^{(0)}$, the above result can be reformulated as follows.
\begin{theorem}\label{nondeg}
If $\ket{\Psi}\in \HN$ is such that its NONs are non-degenerate and saturate a fixed GPC $D_k$, then
\begin{equation}
\mathrm{Supp}_{\mathcal{B}_1}(\kpsi)\subset\{\bd{i}:\ \ket{\bd{i}}\in\mathcal{B}_N{\rm\ and}\ D_k(\bd{n}_i)=0\}.
\end{equation}
\end{theorem}

\subsection{Possibly degenerate occupation numbers}\label{subsec:tangdeg}
As a preliminary point to this subsection we start with an important result which can be viewed as a converse selection rule. Namely, for a fixed GPC, $D_k$, we start from an {\it ansatz} space $\mathcal{A}_k$ which is spanned by configurations of some one-particle orthonormal basis that saturate $D_k$:
\begin{equation}\label{def:ansatz}
\mathcal{A}_k:={\rm Span}\left\{\ket{\bd{i}}:\ D_k(\bd{n}_{\bd i})=0\right\}.
\end{equation}
Note that the one-particle reduced density operator of a generic $\kpsi\in\mathcal{A}_k$ is not necessarily diagonal and it is {\it a priori} not obvious that its NONs can saturate $D_k$ as well. The following theorem asserts that NONs of a generic state actually do saturate $D_k$. However, in order to achieve that, one has to relax their ordering constraints.
\begin{theorem}[Converse selection rule]
Let us fix a GPC, $D_k$, and its corresponding ansatz space $\mathcal{A}_k$ as defined in (\ref{def:ansatz}). For any $\kpsi\in \mathcal{A}_k$ there exists $\mathcal{B}_1'$, an orthonormal basis of natural orbitals $\{\ket{j}\}_{j=1}^d$, such that $\{\ket{{\bd i}}:\ {\bd i}\in\mathrm{Supp}_{\mathcal{B}_1'}(\kpsi)\}\subset\mathcal{A}_k$, i.e.\ the corresponding vector of NONs saturates $D_k$. Moreover, the NONs can be ordered so that $n_i\geq n_j$ whenever $\kappa_k^{(i)}=\kappa_k^{(j)}$ for $i>j$.
\end{theorem}

\begin{proof}
We first show that for any $\kpsi\in \mathcal{A}_k$ the one-particle reduced operator has a block-diagonal form. Namely, for $D_k(\bd n)=\kappa^{(0)}_k+\bd{\kappa}_k\cdotp \bd{n}=0$ we have that
\begin{equation}\label{assertion0}
(\rho_1)_{ij}=0\mathrm{\ if\ } i,j\mathrm{\ are\ such\ that\ }\kappa_k^{(i)}\neq\kappa_k^{(j)},\ i,j>0.
\end{equation}
This in particular means that if all coefficients of vector $\bd{\kappa}_k$ are distinct, then $\rho_1$ is automatically diagonal. To see this, recall that operator $\hat d_k=\sum_j \kappa_k^{(j)}\ket{j}\!\bra{j}$ is a generator local symmetry of $\kpsi$ for any $\kpsi\in\mathcal{A}_k$. By relation (\ref{symmetries-rel}), this automatially implies that $[\rho_1(\kpsi),\hat d_k=0$. Because $\hat d_k$ is diagonal, $\rho_1(\kpsi)$ must be block diagonal as in (\ref{assertion0}).

Furthermore, in order to find NOs of $\kpsi\in\mathcal{A}_k$ it is enough to diagonalise $\rho_1$ within each of its blocks. These blocks concern only such $i,j$-entries of $\rho_1$ that $\kappa_k^{(i)}=\kappa_k^{(j)}$. This means that any $d\times d$ block-diagonal unitary, $U$, which transforms the one-particle orthonormal basis to a basis of NOs while preserving the block-diagonal form of $\rho_1$ is of the form
\begin{equation}
u=\sum_{i,j:\ \kappa_k^{(i)}=\kappa_k^{(j)}}U_{ij}\ket{i}\!\bra{j}.
\end{equation}
Hence, by the preceding discussion, $u^{\otimes N}$ also preserves space $\mathcal{A}_k$, which implies that $\mathrm{Supp}_{\mathcal{B}_1}(\kpsi)\subset\mathcal{A}_k$. Moreover, using such block-diagonal unitaries, one can permute NOs so that the obtained NONs of $\kpsi$ are ordered as in the statement of the theorem.
\end{proof}

Let us next move to the general selection rule, in the case where the NONs can be degenerate. We start with a generalisation of formula (\ref{nondegen-stabiliser}).
\begin{lemma}\label{stabiliser-selection}
Let $\mathcal{B}_1$ be an NO-basis for $\kpsi$. Moreover, denote by $\mk{d}$ the set of one-particle hermitian operators that are diagonal in the basis $\mathcal{B}_1$. Then,
\begin{equation}\label{eq:stab-sel1}
\fl{\mk s}_{\ket{\Psi}}\cap i\mk{d}=\left\{i\sum_{j=1}^d l_i\ket{j}\!\bra{j}:\ \bd{l}\cdotp(\bd{n}_{\bd{i}}-\bd{n})=0{\rm\ for\ all\ }{\bd{i}}\in\mathrm{Supp}_{\mathcal{B}_1}(\kpsi)\right\}.
\end{equation}
Conversely,
\begin{equation}\label{eq:stab-sel2}
\fl\mathrm{Supp}_{\mathcal{B}_1}(\ket{\Psi})\subset\left\{\ket{\bd{i}}\in\mathcal{B}_N:\  \bd{l}\cdotp(\bd{n}_{\bd i}-\bd{n})=0{\rm\ for\ all}\ \hat l_{\mathcal{B}_1}\in\mk{s}_{\ket{\Psi}}\cap\mk{d}\right\}.
\end{equation}
\end{lemma}
\begin{proof}
The proof, in essence, relies on repeating the reasoning from equations (\ref{eq:diag-generator})-(\ref{nondegen-stabiliser}). In particular, for any $\hat L_{\mathcal{B}_1}$ we have
\begin{equation*}
\hat L_{\mathcal{B}_1}\kpsi=\sum_{\bd{i}}c_{\bd{i}}(\bd{l}\cdotp\bd{n}_{\bd{i}})\ket{\bd{i}}.
\end{equation*}
This means that $\hat L_{\mathcal{B}_1}=\sum_{i}l_i\hat n_i$ generates a symmetry of $\kpsi$ if and only if for all $\bd{i}\in \mathrm{Supp}_{\mathcal{B}_1}(\ket{\Psi})$ we have $\bd{l}\cdotp\bd{n}_{\bd{i}}=\lambda$ for some $\lambda\in\RR$. Furthermore, because $\mathcal{B}_1$ is an NO-basis for $\kpsi$, we have $\lambda=\bd{l}\cdotp \bd{n}$ with $\bd n$ being the NON-vector of $\kpsi$. Hence, we have that
\begin{equation}\label{eq:stab-perp}
{\mk s}_{\ket{\Psi}}\cap i\mk{d}=\left(i\mk{d}_{\kpsi}\right)^\perp,
\end{equation}
where
\begin{equation}
\mk{d}_{\kpsi}={\rm Span}\left\{\sum_{j=1}^d \left((\bd{n}_{\bd{i}})_j-n_j\right)\ket{j}\!\bra{j}:\ {\bd{i}}\in\mathrm{Supp}_{\mathcal{B}_1}(\kpsi)\right\}.
\end{equation}
This yields assertion (\ref{eq:stab-sel1}). Conversely, if ${\bd{i}}\in\mathrm{Supp}_{\mathcal{B}_1}(\kpsi)$, then $\sum_{j=1}^d \left((\bd{n}_{\bd{i}})_j-n_j\right)\ket{j}\!\bra{j}\in \mk{d}_{\kpsi}$. By relation (\ref{eq:stab-perp}), we have that $\mk{d}_{\kpsi}=({\mk s}_{\ket{\Psi}}\cap i\mk{d})^\perp$ which yields (\ref{eq:stab-sel2}).
\end{proof}

 Recall that states with degenerate NONs have many NOs. In the remaining part of this section we will make the choice of NOs less ambiguous by requiring them to have a certain additional property. To explain this property, let us first take a closer look at the structure of the local symmetries of a given $\kpsi\in\HN$. Among all generators of local symmetries of $\kpsi$, one can choose a maximal set of operators that commute with each other. Such a set will be denoted by $\mk{s}_\kpsi^{(c)}$. Because all operators from $\mk{s}_\kpsi^{(c)}$ commute with each other, one can find a basis of NOs in which all of them are diagonal. Such NOs will be called {\it adapted}.
\begin{definition}[Adapted NOs]\label{def:adapted}
  Let $\mk{s}_\kpsi^{(c)}$ be a maximal subset of the set of generators of local symmetries of $\kpsi$ that has the property that all operators from $\mk{s}_\kpsi^{(c)}$ commute with each other (in group theory this set is called the Lie algebra of a maximal torus of $S_\kpsi$). NOs $\mathcal{B}_1=\{\ket{i}\}_{i=1}^d$ are called \emph{adapted} if all operators from $\mk{s}_\kpsi^{(c)}$ are diagonal in $\mathcal{B}_1$. In other words, if $ih=i\sum_{i\geq j} H_{ij}\ket{i}\!\bra{j}\in\mk{s}_\kpsi^{(c)}$ then $H_{ij}=0$ whenever $i\neq j$.
\end{definition}

\begin{remark}\label{rem:adapted}
In the adapted basis of NOs we have the orthogonal decomposition
\begin{equation}
{\mk{s}}_{\ket{\Psi}}={\mk{s}}_{\ket{\Psi}}^{(c)}\oplus \left({\mk{s}}_{\ket{\Psi}}^{(c)}\right)^\perp.
\end{equation}
In other words, if $ih\in \mk{s}_\kpsi$, then the diagonal of $ih$ belongs to $\mk{s}_\kpsi$ as well.

To see this, recall that $\mk{s}_\kpsi$ is a compact Lie algebra, hence it decomposes as $\mk{s}_\kpsi=\mk{s}'_\kpsi\oplus \mk{z}_\kpsi$, where $\mk{s}'_\kpsi$ is the semisimple part and $\mk{z}_\kpsi$ is the center. Furthermore, the semisimple part has the following orthogonal decomposition:
\begin{equation}
\mk{s}'_\kpsi=\bigoplus_{\alpha\in R_{+}}\left(\mk{h}_\kpsi^{(\alpha)}\oplus\mk{a}_\kpsi^{(\alpha)}\oplus \mk{b}_\kpsi^{(\alpha)}\right),
\end{equation}
where $R_{+}$ is the set of positive roots of the complexified algebra $\mk{s}'_\kpsi\oplus i\mk{s}'_\kpsi$, $\mk{h}_\kpsi^{(\alpha)}=\RR\left([e_\alpha,e_{-\alpha}]\right)$, $\mk{a}_\kpsi^{(\alpha)}=\RR\left(e_\alpha-e_{-\alpha}\right)$, $\mk{b}_\kpsi^{(\alpha)}=i\RR\left(e_\alpha+e_{-\alpha}\right)$, and $e_\alpha$ is a root operator associated to root $\alpha$. The subalgebra
\begin{equation}
\mk{t}_\kpsi:=\mk{z}_\kpsi\oplus\bigoplus_{\alpha\in R_{+}}\mk{h}_\kpsi^{(\alpha)},
\end{equation}
is a maximal commutative subalgebra of  $\mk{s}_\kpsi$. Denote by $\{\ket{j}\}_{j=1}^d$ a basis that diagonalises the above maximal commutative subalgebra. We will next show that operators from $\mk{a}_\kpsi^{(\alpha)}$ and $\mk{b}_\kpsi^{(\alpha)}$ are purely off-diagonal in this basis. Note first that the commutator of a diagonal matrix with any other matrix necessarily has zero on the diagonal, i.e. $\bra{j}[H,X]\ket{j}=0$ for all $j\in\{1,\dots,d\}$, $H\in \mk{t}_\kpsi$ and $X\in \mk{s}_\kpsi$. Furthermore, for any $\alpha\in R_{+}$ we have $[H,e_\alpha]=\alpha(H) e_\alpha$ and $[H,e_{-\alpha}]=-\alpha(H) e_{-\alpha}$. Hence, for any $X\in \mk{a}_\kpsi^{(\alpha)}$ we find that for all $H\in \mk{t}_\kpsi$
\begin{equation}
\fl 0=\bra{j}[H,X]\ket{j}\propto \bra{j}[H,e_\alpha-e_{-\alpha}]\ket{j}=\alpha(H)\bra{j} e_\alpha+e_{-\alpha}\ket{j}.
\end{equation}
Similarly, for any $X\in \mk{b}_\kpsi^{(\alpha)}$ we find that for all $H\in \mk{t}_\kpsi$
\begin{equation}
\fl 0=\bra{j}[H,X]\ket{j}\propto i\bra{j}[H,e_\alpha+e_{-\alpha}]\ket{j}=i\alpha(H)\bra{j}e_\alpha-e_{-\alpha}\ket{j}.
\end{equation}
It is now straightforward to check that because $\alpha(H)\neq 0$, the above two equations imply that $\bra{j}X\ket{j}=0$ for all $X\in \mk{a}_\kpsi^{(\alpha)}\oplus \mk{b}_\kpsi^{(\alpha)}$.
\end{remark}
\noindent The above notion of adapted NOs allows us to make a connection between the image of $d\mu$ and the structure of a given state $\kpsi$ via lemma \ref{thm:image}. The precise form of this connection is the subject of the following lemma.
\begin{lemma}\label{lemma:image-general}
Let $\mathcal{B}_1=\{\ket{i}\}_{i=1}^d$ be a basis of NOs which is adapted for $\kpsi\in\HN$ and let $\bd{n}$ be the vector of NONs of $\kpsi$. Moreover, denote by $\mk{d}$ the set of one-particle hermitian operators that are diagonal in basis $\mathcal{B}_1$. Then,
\begin{equation}\label{eq:image-general}
\fl\im\ d\mu_{\ket{\Psi}}\cap \mk{d}={\rm Span}\left\{\sum_{j=1}^d \left((\bd{n}_{\bd{i}})_j-n_j\right)\ket{j}\!\bra{j}:\ \bd{i}\in{\rm Supp}_{\mathcal{B}_1}(\kpsi)\right\}.
\end{equation}
\end{lemma}
\begin{proof}
By lemma \ref{thm:image}, for any state $\kpsi\in\HN$, we have
\begin{equation}\label{eq:proof-imagegen}
\im\ d\mu_{\ket{\Psi}}\cap\iota\mk{d}=i\left({\mk s}_{\ket{\Psi}}\right)^\perp\cap\mk{d}=i\left(pr_{\mk{d}}{\mk{s}}_{\ket{\Psi}}\right)^\perp\cap\mk{d},
\end{equation}
where by $pr_{\iota\mk{d}}$ we denote the orthogonal projection on the space of anti-hermitian diagonal matrices. In other words, for any anti-hermitian $A$, $pr_{\iota\mk{d}}(A)$ is the diagonal matrix whose non-zero entries are identical to those of $A$. In the second step of equation (\ref{eq:proof-imagegen}) we have used the fact that the scalar product $\tr(AB)$ depends only on the diagonal part of $A$ if $B$ is diagonal. Indeed, any matrix $A\in {\mk s}_{\ket{\Psi}}$ can be uniquely decomposed as a sum $A=A_{\mk d}+A_{{\mk d}^\perp}$, where $A_{\mk d}:=pr_{\iota\mk{d}}(A)$. The matrix $B\in{\mk d}$ is orthogonal to $A$ iff $\tr\left((A_{\mk d}+A_{{\mk d}^\perp})B\right)=0$. This happens iff $\tr(A_{\mk d}B)=0$. In other words, $B\in \left({\mk s}_{\ket{\Psi}}\right)^\perp\cap{\mk d}$ iff $B\in \left(pr_{\mk{d}}{\mk{s}}_{\ket{\Psi}}\right)^\perp\cap\mk{d}$. Next, we use the fact that the NO-basis in which the above operators are diagonal, is adapted for $\kpsi$. This implies that
\begin{equation}\label{eq:projection-intersection}
pr_{\mk{d}}{\mk{s}}_{\ket{\Psi}}={\mk{s}}_{\ket{\Psi}}\cap\mk{d},
\end{equation}
i.e.\ in an adapted basis on NOs we have that if $A\in {\mk{s}}_{\ket{\Psi}}$, then automatically $A_{\mk d}\in {\mk{s}}_{\ket{\Psi}}$. To see this, recall that ${\mk{s}}_{\ket{\Psi}}\cap\mk{d}$ is precisely ${\mk{s}}_{\ket{\Psi}}^{(c)}$ from definition \ref{def:adapted} written in a basis of adapted NOs. As we explained in remark \ref{rem:adapted}, in an adapted basis we have the orthogonal decomposition
\begin{equation}
{\mk{s}}_{\ket{\Psi}}={\mk{s}}_{\ket{\Psi}}^{(c)}\oplus \left({\mk{s}}_{\ket{\Psi}}^{(c)}\right)^\perp.
\end{equation}
The space $\left({\mk{s}}_{\ket{\Psi}}^{(c)}\right)^\perp$ is the space of matrices with zeros on their diagonals. Hence, taking the diagonal of the matrix $A\in{\mk{s}}_{\ket{\Psi}}$ is the same as projecting $A$ to ${\mk{s}}_{\ket{\Psi}}^{(c)}$.

Equation (\ref{eq:projection-intersection}) applied to (\ref{eq:proof-imagegen}) means that $\im\ d\mu_{\ket{\Psi}}\cap\mk{d}$ is determined by the diagonal matrices that generate local symmetries of $\kpsi$. Finally, by lemma \ref{stabiliser-selection} we have that
the diagonal component of the orthogonal complement of ${\mk s}_{\ket{\Psi}}\cap\mk{d}$ is precisely the right hand side of formula (\ref{eq:image-general}).
\end{proof}

In order to deduce the support of $\kpsi$ from lemma \ref{lemma:image-general} and the knowledge of its NONs, we have to take a closer look at the subtle structure of local symmetries of states with fixed NONs. In the remaining part of this subsection, for simplicity we fix a vector of ordered occupation numbers $\bd{n}$ and a one-particle orthonormal base $\mathcal{B}_1=\{\ket{i}\}_{i=1}^d$.
 Let us distinguish the set of quantum states for which $\mathcal{B}_1$ is a NO-basis and $\bd{n}$ is the vector of NONs.
\begin{equation}\label{def:Mn}
\mathcal{M}_{\bd{n}}:=\left\{\kpsi\in\HN:\ \bra{\Psi}f_i^\dagger f_j\ket{\Psi}=n_i\delta_{ij}\right\}.
\end{equation}
In other words, all states from $\mathcal{M}_{\bd{n}}$ have their one-particle reduced density operators equal to
\begin{equation}
\rho_1^{(\bd{n})}=\sum_{j=1}^dn_j\ket{j}\!\bra{j}.
\end{equation}
The stabiliser of $\rho_1^{(\bd{n})}$ will be denoted by $S_{\bd{n}}$. Importantly, by its definition set $\mathcal{M}_{\bd{n}}$ is $S_{\bd{n}}$-invariant, i.e.\ if $\kpsi\in\mathcal{M}_{\bd{n}}$, then for any one-particle $u$ such that $u\rho_1^{(\bd{n})}u^\dagger=\rho_1^{(\bd{n})}$, we have $u^{\otimes N}\kpsi\in\mathcal{M}_{\bd{n}}$. Finally, note that fixing a one-particle basis does not restrict the generality of the results that will follow, as any state whose ordered NONs are equal to $\bd{n}$ can be brought to $\mathcal{M}_{\bd{n}}$ by a change of its one-particle basis.

The structure of symmetries of states from $\mathcal{M}_{\bd{n}}$ that we are about to review will give us a hierarchy of states from $\mathcal{M}_{\bd{n}}$ and will make precise the notion of a generic state. Mathematically, we will explore the stricture of $\mathcal{M}_{\bd{n}}$ as a {\it stratified symplectic space} \cite{Sjamaar}. This in particular means that space $\mathcal{M}_{\bd{n}}$ can be decomposed into disjoint subsets $\{N_\sigma\}_{\sigma\in\Sigma}$ called strata, for which
\begin{equation}\label{eq:stratification}
\mathcal{M}_{\bd{n}}=\bigsqcup_{\sigma\in\Sigma}N_\sigma.
\end{equation}
In the above expressions, the enumeration is in terms of the discrete set $\Sigma$, which is the set of conjugacy classes of local symmetry groups of states from $\mathcal{M}_{\bd{n}}$. In other words, two states $\kpsi$ and $\ket{\Psi'}$ belong to the same stratum if and only if $S_{\kpsi}=u S_{\ket{\Psi'}}u^\dagger$ for some $u\in S_{\bd n}$. Let us next motivate this construction. If two states can be transformed into each other by a matrix $u\in S_{\bd{n}}$, their local symmetry groups are conjugate, i.e.
\begin{equation}
S_{u^{\otimes N}\kpsi}=uS_{\kpsi} u^\dagger.
\end{equation}
In other words, $S_{u^{\otimes N}\kpsi}$ and $S_{\kpsi}$ are in the same conjugacy class $\sigma$. However, the converse is not true -- two states may have local symmetry groups from the same conjugacy class without being unitarily equivalent. Nevertheless, they are arranged to form a single stratum. The stratification (\ref{eq:stratification}) has three important properties.

\begin{theorem}\label{stratification-properties}
Let $\{N_\sigma\}_{\sigma\in\Sigma}$ be the symplectic stratification of $\mathcal{M}_{\bd{n}}$. The following properties hold \cite{Sjamaar}.
\begin{enumerate}
\item If $\overline{N_\sigma}\cap N_{\sigma'}\neq\emptyset$, then $\overline{N_\sigma}\supset N_{\sigma'}$. This means that strata can be partially ordered with respect to the relation of inclusion, i.e.\ we say that $N_\sigma$ is bigger than $N_{\sigma'}$ if and only if $\overline{N_\sigma}\supset N_{\sigma'}$.
\item If $\kpsi\in N_\sigma$ and $\ket{\Psi'}\in N_{\sigma'}$ with $\overline{N_\sigma}\supset N_{\sigma'}$, their stabilisers are related by
\begin{equation}uS_{\kpsi}u^\dagger\subset S_{\ket{\Psi'}}\end{equation}
for some $u\in S_{\bd n}$.
\item There exists a unique maximal stratum, $N_{max}$, which is of maximal dimension and is open and dense in $\mathcal{M}_{\bd{n}}$.
\end{enumerate}
\end{theorem}
By calling a state \emph{generic} we mean that it belongs to stratum $N_{max}$. As we will show next, the uniqueness of the stratum $N_{max}$ together with lemma \ref{lemma:image-general} implies the existence of a selection rule which is universal for all states from $\mathcal{M}_{\bd{n}}$.

\begin{definition}[$\sigma$-ansatz]\label{def:sigma-ansatz}
  Let $\kpsi\in N_{\sigma}\subset \mathcal{M}_{\bd{n}}$ be such that $\mathcal{B}_1$ is an adapted basis of NOs. A $\sigma$-ansatz space for $\mathcal{M}_{\bd{n}}$, $\mathcal{A}_{\bd n}^{(\sigma)}$ is defined as the set of all configurations whose (shifted) NON-vectors are perpendicular to all diagonal operators from $\mk{s}_{\kpsi}$.
\begin{equation}
\mathcal{A}_{\bd n}^{(\sigma)}:=\left\{\ket{\bd{i}}\in\mathcal{B}_N:\  \bd{l}\cdotp(\bd{n}_{\bd i}-\bd{n})=0{\rm\ for\ all}\ \hat l_{\mathcal{B}_1}\in\mk{s}_{\kpsi}\cap\iota\mk{d}\right\}.
\end{equation}
\end{definition}
Note that for different $\sigma$'s the $\sigma$-anzatz $\mathcal{A}_{\bd n}^{(\sigma)}$ can change as $\mk{s}_{\kpsi}\cap\iota\mk{d}$ is determined by $\sigma$.
\begin{remark}\label{obs:ansatz}
  Importantly, as the local symmetry groups of all states from $N_{\sigma}$ are conjugate to each other, $\sigma$-ansatzes for different $\kpsi\in N_{\sigma}\subset \mathcal{M}_{\bd{n}}$ are the same up to permutation of elements of the chosen adapted basis of NOs.
\end{remark}
\begin{remark}\label{rem:sigma-ansatz}
Lemma \ref{stabiliser-selection} applied to the case of an adapted basis of NOs implies that if $\kpsi\in N_{\sigma}\subset \mathcal{M}_{\bd{n}}$ and $\mathcal{B}_1$ is an adapted basis of NOs for $\kpsi$, then
\begin{equation}
\{\ket{{\bd i}}:\ {\bd i}\in\mathrm{Supp}_{\mathcal{B}_1}(\kpsi)\}\subset \mathcal{A}_{\bd n}^{(\sigma)}.
\end{equation}
\end{remark}

\begin{theorem}[Maximal ansatz is universal]\label{thm:maximal-ansatz}
For any $\kpsi\in \mathcal{M}_{\bd{n}}$ there exists an operator $u\in S_{\bd n}$ such that
\begin{equation}\label{eq:max-ansatz}
\{\ket{{\bd i}}:\ {\bd i}\in\mathrm{Supp}_{\mathcal{B}_1}(u^{\otimes N}\kpsi)\}\subset\mathcal{A}_{\bd n}^{(max)},
\end{equation}
where $\mathcal{A}_{\bd n}^{(max)}$ is any $\sigma$-ansatz corresponding to the maximal stratum $N_{max}\subset\mathcal{M}_{\bd{n}}$.
\end{theorem}
\begin{proof}
Let $\kpsi\in N_{max}$ and $\ket{\Psi'}\in N_\sigma$ for some stratum $\sigma$. Moreover, assume that $\mathcal{B}_1$ is an adapted basis of NOs for $\ket{\Psi'}$ (this can be always satisfied by transforming $\ket{\Psi'}\to w\ket{\Psi'}$ for proper $w\in S_{\bd n}$). By point (ii) of theorem \ref{stratification-properties}, we have that there exists $u\in S_{\bd{n}}$ for which
\begin{equation}\label{eq:strat-def}
uS_{\kpsi}u^\dagger\subset S_{\ket{\Psi'}}.
\end{equation}
Equation (\ref{eq:strat-def}) on the level of generators reads as
\begin{equation}
{\mk s}_{u^{\otimes N}\kpsi}\subset {\mk s}_{\ket{\Psi'}}.
\end{equation}
In particular, we have that ${\mk s}_{u^{\otimes N}\kpsi}\cap\iota\mk{d}\subset {\mk s}_{\ket{\Psi'}}\cap\iota\mk{d}$. Hence, by lemma \ref{stabiliser-selection}
\begin{eqnarray*}
\fl\left\{\ket{\bd{i}}\in\mathcal{B}_N:\  \bd{l}\cdotp(\bd{n}_{\bd i}-\bd{n})=0{\rm\ for\ all}\ \hat l_{\mathcal{B}_1}\in\mk{s}_{u^{\otimes N}\kpsi}\cap\iota\mk{d}\right\} \supset \\ \supset \left\{\ket{\bd{i}}\in\mathcal{B}_N:\  \bd{l}\cdotp(\bd{n}_{\bd i}-\bd{n})=0{\rm\ for\ all}\ \hat l_{\mathcal{B}_1}\in\mk{s}_{\ket{\Psi'}}\cap\iota\mk{d}\right\}.
\end{eqnarray*}
Note that relation (\ref{eq:strat-def}) implies that $\mathcal{B}_1$ must also be an adapted basis of NOs for $u^{\otimes N}\kpsi$, hence by remark \ref{obs:ansatz} for every maximal ansatz $\mathcal{A}_{\bd n}^{(max)}$ there exists a permutation matrix $v\in S_{\bd n}$ such that
\begin{eqnarray*}
\left\{\ket{\bd{i}}\in\mathcal{B}_N:\  \bd{l}\cdotp(\bd{n}_{\bd i}-\bd{n})=0{\rm\ for\ all}\ \hat l_{\mathcal{B}_1}\in\mk{s}_{u^{\otimes N}\kpsi}\cap\iota\mk{d}\right\}=v\mathcal{A}_{\bd n}^{(max)}.
\end{eqnarray*}
Hence,
\begin{eqnarray*}
\left\{\ket{\bd{i}}\in\mathcal{B}_N:\  \bd{l}\cdotp(\bd{n}_{\bd i}-\bd{n})=0{\rm\ for\ all}\ \hat l_{\mathcal{B}_1}\in\mk{s}_{v^{\otimes N}\ket{\Psi'}}\cap\iota\mk{d}\right\}\subset\mathcal{A}_{\bd n}^{(max)}.
\end{eqnarray*}
Finally, assertion (\ref{eq:max-ansatz}) follows from directly lemma \ref{stabiliser-selection} applied to state $v^{\otimes N}\ket{\Psi'}$.
\end{proof}
In the light of theorem \ref{thm:maximal-ansatz}, finding the selection rule in the general case of possibly degenerate occupation numbers boils down to determining the maximal ansatz. As we will next show, this is possible under some additional technical assumptions that concern the distribution of vertices of the Pauli hypercube relatively to the GPCs that are saturated by the given vector $\bd n$.

Let us start with the case when occupation vector $\bd n$ saturates exactly one GPC, $D_k$. Pick a state $\kpsi\in \mathcal{M}_{\bd n}\cap \mathcal{A}_k$ (see definition \ref{def:ansatz}) such that $\mathcal{B}_1$ is an adapted basis of NOs for $\kpsi$. We necessarily also have that $\kpsi$ belongs to some stratum $N_\sigma\subset \mathcal{M}_{\bd n}$. Because $\kpsi\in\mathcal{A}_k$, $\kpsi$ is stabilised by operator
\[\hat K_{\mathcal{B}_1}=\sum_{j=1}^d\kappa_k^{(j)}\hat n_j.\]
Moreover, because $\mathcal{B}_1$ is an adapted basis of NOs for $\kpsi$, by remark \ref{rem:sigma-ansatz} we have
\begin{equation}
\mathrm{Supp}_{\mathcal{B}_1}(\kpsi)\subset \mathcal{A}_{\bd n}^{(\sigma)}\subset \mathcal{A}_k.
\end{equation}
Using arguments analogous to the ones used in the proof of theorem \ref{thm:maximal-ansatz} one can show that any generic state $\ket{\Psi'}\in N_{max}\subset \mathcal{M}_{\bd n}$ can be transformed via $u\in S_{\bd n}$ to a state $\ket{\tilde\Psi}$ for which $\mathcal{B}_1$ is an adapted basis of NOs and $S_{\ket{\tilde\Psi}}\subset S_{\kpsi}$. In particular, because $\mathcal{B}_1$ is an adapted basis of NOs for both $\kpsi$ and $\ket{\tilde\Psi}$, all operators from the maximal commutative subalgebra of $s_{\ket{\tilde\Psi}}$ which is diagonal in $\mathcal{B}_1$ belong to the corresponding maximal commutative subalgebra of $s_{\kpsi}$. Note that by remark \ref{obs:ansatz} we have a freedom of acting on $\ket{\tilde\Psi}$ with any permutation matrix $v\in S_{\bd n}$. Two following mutually exclusive cases are possible.
\begin{enumerate}
\item  There exists a permutation matrix $v\in S_{\bd n}$ such that $\hat k_{\mathcal{B}_1}=\sum_{j=1}^d\kappa_k^{(j)}\ket{j}\!\bra{j}\in vs_{\ket{\tilde\Psi}}v^\dagger$. Then we have that ansatz $\mathcal{A}_k$ leads to the universal selection rule, i.e.\ $\mathcal{A}_{\bd n}^{(max)}\subset \mathcal{A}_k$.
\item For all permutation matrices $v\in S_{\bd n}$ we have $\hat k_{\mathcal{B}_1}\notin vs_{\ket{\tilde\Psi}}v^\dagger$. Then the maximal commutative subalgebra of $s_{\kpsi}$ contains $\hat k_{\mathcal{B}_1}$ and another operator $\hat l_{\mathcal{B}_1}$ which is linearly independent from  $\hat k_{\mathcal{B}_1}$ and is not proportional to $v\hat k_{\mathcal{B}_1}v^\dagger$ for all permutation matrices $v\in S_{\bd n}$. This means that vector $\bd{l}$ may define universal ansatz which is different from $\mathcal{A}_k$, i.e.
\begin{equation}
\mathcal{A}_{\bd n}^{(max)}\subset{\rm Span}\left\{\ket{\bd{i}}\in\mathcal{B}_N:\  \bd{l}\cdotp(\bd{n}_{\bd i}-\bd{n})=0\right\}\neq\mathcal{A}_k.
\end{equation}
\end{enumerate}
The following technical assumption \ref{assumption} allows us to exclude the above case (ii). More specifically,  we want to exclude the possibility of the existence of hyperplanes spanned by vertices of the Pauli hypercube which are different than $D_k$ and all its relevant reflections. This can be achieved by incorporating the following combinatorial procedure. Let us first introduce the group $\Pi_{\bd{n}}$ of permutation $\pi$ which leaves $\bd{n}$ invariant,
\begin{equation}\label{perm}
\Pi_{\bd{n}} \equiv \{\mbox{permutations}\,\pi\,|\,\pi(\bd{n})= \bd{n}\}\,.
\end{equation}
For instance, if $\bd{n}$ has only a twofold degeneracy $n_j=n_{j+1}$ the group $\Pi_{\bd{n}}$ consists of only two elements, $\Pi_{\bd{n}}=\{\mathds{1},\pi_{j,j+1}\}$. In the general case group $\Pi_{\bd{n}}$ is generated by transpositions $\pi_{j,j+1}$ for all $j$ such that $n_j=n_{j+1}$
\begin{equation}
\Pi_{\bd{n}}=\left\langle\pi_{j,j+1}:\ n_j=n_{j+1}\right\rangle.
\end{equation}
Let us denote by $D_k\circ \pi_{j,j+1}$ the reflection of GPC $D_k$ with respect to permutation $\pi_{j,j+1}$. Then, halfspace $\{\bd n:\ D_k\circ \pi_{j,j+1}(\bd n)\geq 0\}$ is the reflection of halfspace $\{\bd n:\ D_k(\bd n)\geq 0\}$ with respect to hyperplane given by equation $n_{j}=n_{j+1}$.
\begin{assumption}\label{assumption}
For every GPC, $D_k$, let $\bd n$ be a generic point that saturates $D_k$ and has degeneracy $n_j=n_{j+1}$. It is not possible to find vector $\bd{l}$ with the following properties (see Fig. \ref{fig:assumption}).  For $V_{\bd l}:=\{\bd{n'}:\ \bd{l}\cdotp \bd{n'}=0\}$:
\begin{enumerate}
\item $V_{\bd l}=\mathrm{Span}\left\{\bd{n}_{\bd i}:\ket{{\bd i}}\in\mathcal{B}_N\ \mathrm{and}\ \bd{l}\cdotp(\bd{n}_{\bd i}-\bd{n})=0\right\}$,
\item $\bd{l}$ is not proportional to $\pi_{j,j+1}(\bd{\kappa}_k)$ or $\bd{\kappa}_k$.
\item Set $\mathrm{Conv}\left\{\bd{n}_{\bd i}:\ket{{\bd i}}\in\mathcal{B}_N\ \mathrm{and}\ \bd{l}\cdotp(\bd{n}_{\bd i}-\bd{n})=0\right\}$ is contained in the region of admissible one-particle spectra in the vicinity of point $\bd n$.
\end{enumerate}
\end{assumption}
\begin{figure}[h]
\centering
\includegraphics[width=0.6\textwidth]{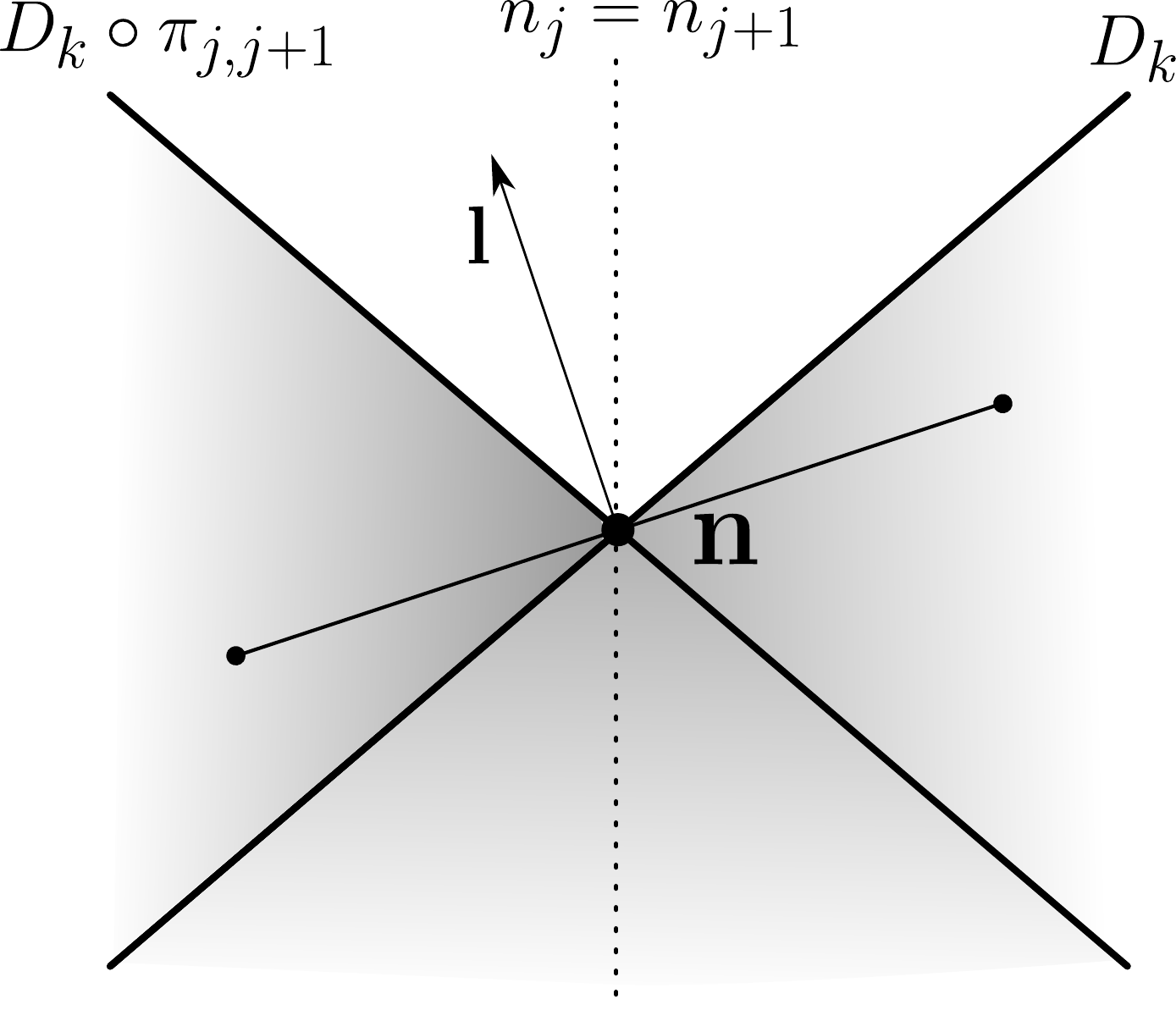}
\caption{Schematic supplementary figure for assumption \ref{assumption}. Solid line segment depicts convex hull of vertices of the Pauli hypercube that span space $V_{\bd l}$. Light grey color shows region where the spectral polytope and the reflected spectral polytope are contained.}
\label{fig:assumption}
\end{figure}
The validity of this assumption has been confirmed by us numerically in all cases where the GPCs are known, i.e.\ for $N\leq 5$ and $d\leq 11$. More specifically, for all GPCs we verified that generic $\bd n$ from the intersection of a given GPC with a number of hyperplanes $n_j=n_{j+1}$ satisfies assumption \ref{assumption}. Therefore, we have hardly any doubts that this holds in general. However, proving this fact in a straightforward way is a challenging combinatorial problem. Finally, let us point out that proving the validity of assumption \ref{assumption} for every GPC and every possible degeneracy would allow us to strengthen the universal ansatz theorem.
\begin{corollary}
Let $\bd n$ saturate multiple GPCs. Under assumption \ref{assumption}, theorem \ref{thm:maximal-ansatz} implies that there exists one of the saturated GPCs, $D_k$, such that for every state $\kpsi\in\mathcal{M}_{\bd n}$ there exists a transformation $u\in S_{\bd n}$ such that $u^{\otimes N}\kpsi\in\mathcal{A}_k$.
\end{corollary}

\section{Global implications of extremal local quantum information}\label{sec:general}
%tbd (Adam, Tomek)\\
%Here we discuss non-fermionic settings.
%\begin{enumerate}
%\item Define briefly what we mean by single-body quantum marginal problem (referring to local groups, invariance,\ldots, recall the mathematical structure required for Klyachko's work: K\"ahler manifolds, moment map \ldots)
%\item Explain briefly why the 2-body N-representability problem for fermions and the 1-body N-representability problem for
%hard-core bosons are not single body QMP
%\item Briefly explain how the crucial objects of the fermionic setting and the respective proofs translate to those
%for non-fermionic settings
%\item Illustration: N-qubits and 1-body N-representability problem for bosons
%\end{enumerate}

In this section we will briefly discuss non-fermionic settings. The single-body quantum marginal problem appears naturally in quantum information theory, where systems of distinguishable particles are considered. The map $\mu$ assigns to a state $\ket\Psi$ its $1$-particle reduced density matrices. The solution of one-body quantum marginal problem is obtained {\it mutatis mutandis} to the case of fermions as all the essential mathematical structures are also present for distinguishable particles, i.e.\ $\mathbb{P}(\mathcal{H})$ is a K\"ahler manifold and $\mu$ arises as the momentum map of the local unitary action on the space of states (see
\cite{S18,MS18,TomekDoub,vergne2017inequalities,MS15,OSS14,SOK14,MOS13,ST13,HKS13,SWK13,walter2013entanglement,SOK12,SK11,SHK11} for more examples of the usage of geometric techniques in quantum information). A similar observation applies to bosons. In the following we fully characterise selection rules for the system of qubits and bosons.

We consider a multipartite quantum system consisting of $r$ distinguishable subsystems $S_1,\ldots, S_r$. Moreover, we
denote the reduced density operators by $\rho_{S_j}$ and assume that all respective local Hilbert spaces $\mathcal{H}_{S_j}$ have the same dimension $d$. The latter could always be achieved by embedding the possibly smaller dimensional $\mathcal{H}_{S_j}$ into larger spaces of the dimension $d\equiv \max_j\big(\mbox{dim}(\mathcal{H}_{S_j})\big)$. The quantum marginal constraints,
\begin{equation}\label{qmc}
D_k\big(\bd{n}_{S_1},\ldots, \bd{n}_{S_r}\big)\equiv \kappa_k^{(0)}+\sum_{j=1}^r \bd{\kappa}_k \bd{\cdot} \bd{n}_{S_j}\geq 0\,,
\end{equation}
represent the necessary and sufficient conditions for the compatibility of given non-increasingly ordered local spectra $\bd{n}_{S_j}\equiv \big(n_{S_j,i}\big)_{i=1}^d\equiv \mbox{spec}^\downarrow\!\left(\rho_{S_j}\right)$
to a pure total state.
In complete analogy to the case of identical fermions (as discussed in the previous sections), the saturation of a quantum marginal constraint \eqref{qmc} implies  a \emph{selection rule}. To be more specific, whenever $D\big(\bd{n}_{S_1},\ldots, \bd{n}_{S_r}\big)=0$ for some constraint $D\geq 0$, there exist at least one family of local bases $\mathcal{B}_{S_j}\equiv \{\ket{i}_{\!\mathcal{S}_j}\}_{i=1}^d$ of eigenstates of $\rho_{S_j}$, $j=1,2,\ldots,r$, such that the corresponding multipartite state $\ket{\Psi}$ fulfills
\begin{equation}\label{pin3}
\hat{D} \ket{\Psi} \equiv D\Big[\big(\hat{n}_{\mathcal{S}_1}^{(i)}\big)_{i=1}^d,\ldots,\big(\hat{n}_{\mathcal{S}_N}^{(i)}\big)_{i=1}^d\Big]\,\ket{\Psi}=0\,.
\end{equation}
Here, $\hat{n}_{\mathcal{S}_1}^{(i)}\equiv \ket{i}_{\!\mathcal{S}_1 \mathcal{S}_1}\!\bra{i}\otimes \mathbb{1}_{\mathcal{S}_2}\otimes \ldots \otimes \mathbb{1}_{\mathcal{S}_N}$ ($\hat{n}_{\mathcal{S}_j}^{(i)}$ is defined analogously for $j>1$), we suppressed the dependence of $\hat{D}$ on the local bases $\mathcal{B}_{S_j}$ of (possibly non-unique) eigenstates $\ket{i}_{\mathcal{S}_1}$ of $\rho_{S_j}$ and we recall that those local bases are unique as long as each of the local spectra $\bd{n}_{S_j}$ is non-degenerate. To work out the consequences of \eqref{pin3}, we express $\ket{\Psi}$ with respect to the tensor product states built from the local bases $\mathcal{B}_{S_1},\ldots, \mathcal{B}_{S_r}$,
\begin{equation}\label{Psiself}
\ket{\Psi} = \sum_{i_1,\ldots,i_r=1}^d c_{i_1,\ldots,i_r} \,\ket{i_1}_{\!\mathcal{S}_1}\!\otimes \ldots \otimes \ket{i_r}_{\!\mathcal{S}_r}\,.
\end{equation}
The expansion \eqref{Psiself} is the analogue of the natural orbital expansion of fermionic quantum states (recall eq.~(8) in Part I \cite{Pin1}).
Since the states $\ket{i_1}_{\!\mathcal{S}_1}\!\otimes \ldots \otimes \ket{i_r}_{\!\mathcal{S}_r}$ are the eigenstates of $\hat{D}$ with (integer) eigenvalues $D(\bd{n}_{S_1,\bd{i}},\ldots, \bd{n}_{S_r,\bd{i}})= \kappa^{(0)}+\sum_{j=1}^r \kappa^{(i_j)}$, \eqref{pin3} implies that only those configurations $\bd{i}\equiv (i_1,\ldots,i_r)$ may contribute to the self-consistent expansion \eqref{Psiself} whose unordered spectra $(\bd{n}_{S_1,\bd{i}},\ldots, \bd{n}_{S_r,\bd{i}})$,
\begin{equation}\label{NONconf}
\left(\bd{n}_{S_j,\bd{i}}\right)_k= \delta_{i_j,k}\,,
\end{equation}
lie on the hyperplane defined by $D\equiv 0$. In complete analogy to the case of fermions (recall Corollaries 7,11 in Part I \cite{Pin1}), a stronger selection rule may apply in case of degenerate spectra. In the following, we illustrate those structural implications of pinning for non-fermionic single-body quantum marginal problems.
%the eigenvector of $\rho_{\mathcal{S}_1}$ corresponding to its $i$-th largest eigenvalue $\lambda_{\mathcal{S}_1}^{(i)}$.

\subsection{Examples}\label{subsec:examples}
\subsubsection{$r$ qubits}
One prominent example for a non-fermionic single-body quantum marginal problem is the one of $r$ qubits.
Their reduced density operators $\rho_{S_1},\ldots,\rho_{S_r}$ are compatible to a common $r$-qubit pure state $\ket{\Psi}$
if and only if their spectra $\bd{n}_{S_1},\ldots,\bd{n}_{S_r}$ fulfill the following polygonal inequalities \cite{Higushi}
\begin{equation}\label{Hig}
D_i(\bd{n}_{S_1},\ldots,\bd{n}_{S_r})\equiv -n_{S_i}^{(2)}+ \sum_{j\neq i} n_{S_j}^{(2)} \geq 0\,,
\end{equation}
for all $i$. Here, $n_{S_i}^{(2)}$ denotes the smaller eigenvalue of $\rho_{S_i}$ which fixes the spectrum of $\rho_{S_i} \equiv n_{S_i}^{(1)} \ket{1}_{\!\mathcal{S}_i \mathcal{S}_i}\!\bra{1}+n_{S_{i}}^{(2)} \ket{2}_{\!\mathcal{S}_i \mathcal{S}_i}\!\bra{2}$ via the normalization $\mbox{Tr}[\rho_{S_i}]=n_{S_i}^{(1)}+n_{S_i}^{(2)}=1$.

In case a quantum marginal constraint \eqref{Hig} is saturated, e.g., $D_1(\bd{n})=0$, there exist local bases $\mathcal{B}_{S_j}$ such that the corresponding  $r$-qubit  state $\ket{\Psi}$ lies according to \eqref{pin3} in the zero-eigenspace of the respective $\hat{D}_1$-operator,
\begin{equation}
\hat{D}_1\ket{\Psi} = 0\,.
\end{equation}
By expressing $\ket{\Psi}$ in the self-consistent expansion \eqref{Psiself}, we conclude that only the configurations $(1,\ldots,1)$ and $(2,2,1,1,\ldots)$, $(2,1,2,1,1,\ldots)$, \ldots, $(2,1,\ldots,1,2)$ may contribute to $\ket{\Psi}$. Hence, pinning by one quantum marginal constraint \eqref{Hig} would reduce the number of contributing configurations from $2^r$ to just $r$. It is also interesting that for qubits all the difficulties described in Section \ref{subsec:tangdeg} are not present. The rest of this section explains this phenomenon.

\begin{lemma}\label{nontrivial1}
If for every GPC $D_k$ and its reflection $D_k^\prime$ all the vertices $\bd{n_i}$ of the Pauli hypercube satisfy $D_k(\bd{n_i})\geq 0$ and $D_k^\prime(\bd{n_i})\geq 0$ then the selection rule is given by $\mathcal{A}_k$.
\end{lemma}

We will show next that conditions of Lemma \ref{nontrivial1} are satisfied for $r$-qubit system.

%The equations describing $\Delta_\mathcal{H}=\mu(\mathbb{P}(\mathcal{H}))\cap\mathfrak{t}_+$ for $L$ qubits are given by:

%\begin{eqnarray}
%\forall_{k}\,\,\, \sum_{j\neq k}\lambda_{j}-\lambda_k\leq \frac{1}{2}L-1\label{nier1}\\
%\forall_{k} \lambda_k\in [0, \frac{1}{2}].\label{nier2}
%\end{eqnarray}

\begin{lemma}\label{nontrivial}
Assume that $n_{S_l}^{(2)}=\frac{1}{2}$. Then for $k\neq l$ all the inequalities $\sum_{j\neq k}n_{S_j}^{(2)}-n_{S_k}^{(2)}\geq 0$ are automatically satisfied. The only nontrivial inequality is $\sum_{j\neq l}n_{S_j}^{(2)}\geq \frac{1}{2}$ and its saturation defines a codimension $2$ edge of $\mathcal{P}$.
\end{lemma}
%\begin{proof}
%
%Assume $n_l=\frac{1}{2}$. Next, assuming $k\neq l$ the lhs of (\ref{nier1}) belongs to $[-\frac{1}{2},\frac{1}{2}L-\frac{3}{2}]$ which is always smaller than $\frac{1}{2}L-1$. For $k=l$ the lhs of (\ref{nier1}) belongs to $[0,\frac{1}{2}L-\frac{1}{2}]$ thus inequality (\ref{nier1}) is not satisfied for all $\lambda$'s satisfying (\ref{nier2}).
%\end{proof}

Using Lemma $\ref{nontrivial}$, for our further calculations, it is enough to consider one of the inequalities (\ref{Hig}), for example the one with $i=1$. The reflection of the hyperplane  $\sum_{j=2}^Ln_{S_j}^{(2)}-n_{S_1}^{(2)}= 0$ along $n_{S_1}^{(2)}=\frac{1}{2}$ gives $\sum_{j=2}^Ln_{S_j}^{(2)}+n_{S_1}^{(2)}= 1$. Using Lemma \ref{nontrivial1} we know that vertices of Pauli hypercube that can lead to non-standard selection rules are those that satisfy:
\begin{eqnarray}
\sum_{j=2}^Ln_{S_j}^{(2)}-n_{S_1}^{(2)}< 0,\,\,\mathrm{or}\label{w1}\\
\sum_{j=2}^Ln_{S_j}^{(2)}+n_{S_1}^{(2)}<1.\label{w2}
\end{eqnarray}

%Recall that weights are of the form $(\pm\frac{1}{2},\ldots,\pm\frac{1}{2})$. The weight containing $k$ coordinates equal to $-\frac{1}{2}$ will be called $k$-excited.

\begin{lemma}\label{weights}
The only vertex of the Pauli hypercube that satisfies (\ref{w1}) is $(1,0,\ldots,0)$ and the only vertex of the Pauli hypercube that satisfies (\ref{w2}) is $(0,\ldots,0)$.
\end{lemma}

\begin{proof}
It is easy to verify that both vertices indeed satisfy desired conditions. To show that these are complete assume that a vertex $w$ has $k$ zeros. There are two cases to consider for each inequality, i.e.\ $n_{S_1}^{(2)}=0$ or $n_{S_1}^{(2)}=1$. We start with inequality (\ref{w1}).

\begin{eqnarray}
n_{S_1}^{(2)}=0,\,\,\,L-k<0\label{w11}\\
n_{S_1}^{(2)}=1,\,\,\,L-k<2,\label{w12}
\end{eqnarray}
One can easily see that (\ref{w11}) is never satisfied (for $k\geq 0$) and the only $k>0$ that satisfies (\ref{w12}) is $k=L-1$ . Similarly for inequality (\ref{w2}) we get:

\begin{eqnarray}
n_{S_1}^{(2)}=0,\,\,\,L-k<1,\label{w21}\\
n_{S_1}^{(2)}=1,\,\,\,L-k<1,\label{w22}
\end{eqnarray}
One can easily see that (\ref{w22}) is never satisfied (for $k<L$) and (\ref{w21}) is satisfied only when $k=L$.

\end{proof}
Note that the vertex $(1,0,\ldots,0)$ is reflection of the vertex $(0,\ldots,0)$ along $n_{S_1}^{(2)}=\frac{1}{2}$. The conclusion from Lemma \ref{weights} is:
\begin{theorem}
For a system of $r$ qubits all selection rules are given by $\mathcal{A}_k$.
\end{theorem}

\begin{proof}
By Lemma \ref{weights}, the light grey and white regions in Figure \ref{fig:example-weight} do not contain any vertices of the Pauli hypercube. Thus all the vertices satisfy conditions of lemma \ref{nontrivial1} (are in the dark grey region) and the result follows.
\end{proof}
\begin{figure}[h]
\centering
\includegraphics[width=0.6\textwidth]{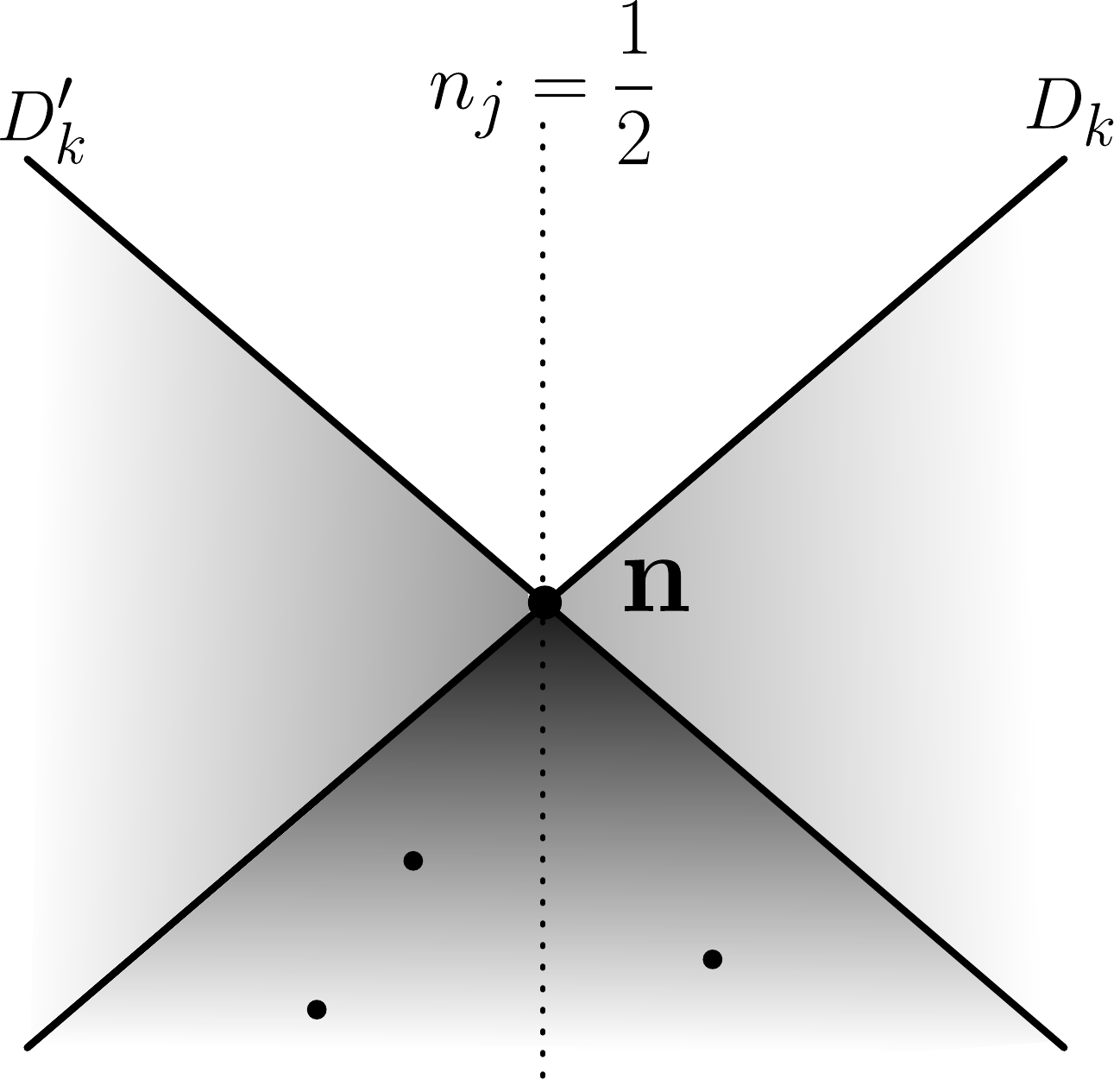}
\caption{Location of vertices of the Pauli hypercube for qubits. Dark grey region corresponds to $D_k(\bd{n_i})\geq 0$ and $D_k^\prime(\bd{n_i})\geq 0$.}
\label{fig:example-weight}
\end{figure}

\subsubsection{$N$ bosons}
In contrast to $N$ identical fermions, the single-body quantum marginal problem is trivial for $N$ identical bosons. Indeed, for any vector $\bd{n}\equiv (n_j)_{j=1}^d$ of NONs respecting the trivial constraints $n_j \geq 0$ and normalization $n_1+\ldots+n_d=N$, one can find a corresponding $N$-boson quantum state $\ket{\Psi}\in \mathcal{H}_N^{(b)}\equiv \mbox{Sym}^N[\Ho]$, e.g.,
\begin{equation}
\ket{\Psi}= \frac{1}{\sqrt{N}}\sum_{j=1}^d \sqrt{n_j} \,\ket{j}\otimes \ldots \otimes \ket{j}\,.
\end{equation}
Here, $\{\ket{j}\}$ are some orthonormal states. Nonetheless, the selection rule based on \eqref{pin3} applies also to the saturation of the trivial constraints $n_j\geq 0$, just implying that the respective natural orbital $\ket{j}$ does not contribute at all to the self-consistent expansion of $\ket{\Psi}$ in bosonic configurational states $\ket{i_1,\ldots,i_N}_b$ built up from the respective natural orbitals. If even all except one constraint are saturated, we have $n_1=N$ and thus the presence of a complete Bose-Einstein condensate, $\ket{\Psi}=\ket{1,\ldots,1}_b\equiv \ket{1}^{\otimes^N}$.

\subsubsection{$N$ hard-core bosons}
The bosonic single-body quantum marginal problem gets highly non-trivial if we restrict the $N$-boson Hilbert space $\mathcal{H}_N^{(b)}$ to the one of $N$ hard-core bosons. To be more specific, after introducing the orthonormal basis $\{\ket{\chi_j}\}_{j=1}^d$ of hard-core/lattice site states, $\mathcal{H}_N^{(b)}$ is restricted to $\mathcal{H}_N^{(hcb)}$ by skipping all configurations multiply occupying any of those lattice sites. This gives rise to a well-defined quantum marginal problem. It seeks the one-particle reduced density operators $\rho_1$ that are representing a quantum state in $\mathcal{H}_N^{(hcb)}$. Its complete solution would allow one to efficiently determine the ground state energies of all quantum systems of identical bosons which interact only by hard-core interaction (including all possible one-particle terms in the Hamiltonian). Unfortunately, the space $\mathcal{H}_N^{(hcb)}\leq \mathcal{H}_N^{(b)}$ is not invariant under rotations of the one-particle Hilbert space $\mathcal{H}_1$ and the formalism by Klyachko and its solution do not apply. In particular, the solution set of $N$-representable one-particle reduced density operators $\rho_1$ is described by conditions involving the natural orbitals as well \cite{FTthesis}. Although this set takes a less preferable form, one could try to find outer approximations to it in analogy to the fundamentally important two-body $N$-representability problem for fermions (see, e.g., Refs.~\cite{Coleman,Mazz04,Mazz12,Mazz16}). One prominent outer approximation is given by the exclusion principle analogue: The largest possible occupation number that can be found in a system of $N$ hard-core bosons on $d$ lattice sites is given by \cite{CSHCB}
\begin{equation}
n_1 \leq N_{max}\equiv \frac{N}{d}(d-N+1)\,.
\end{equation}
Saturation of this universal upper bound on the degree of condensation of hard-core bosons implies that there is one natural orbital, $\ket{1}$, of $\rho_1$ which is maximally unbiased  with respect to the lattice site basis $\{\ket{\chi_j}\}$, $\langle \chi_j \ket{1}=1/\sqrt{d}$ for all $j$, and the corresponding quantum state is given by (up to some phases which could be transformed away) is maximally delocalized,
$\ket{\Psi} \propto \sum_{i_1<i_2<\ldots <i_N} \ket{\chi_{i_1},\ldots,\chi_{i_N}}_b$ (see Ref.~\cite{CSHCB} for more details).

\section{Summary}\label{sec:conl}
Extension of our results to general non-fermionic multipartite quantum systems reveals that extremal single-body information has always strong implications for the multipartite quantum state. In that sense, we confirm that pinned quantum systems define new physical entities since their response to adiabatic external perturbations has to be restricted to the corresponding polytope facet. Our approach also establishes a beautiful link between representation theory, geometry and the extremal single-body information. In our work we distinguish two scenarios. The first one concerns nondegenrate NONs and the main result is Theorem  \ref{nondeg}. The proof of this theorem is based on the fundamental property of the momentum map which relates the image $\im d\mu_{\ket{\Psi}}$ with the Lie algebra of the stabiliser of $\ket{\Psi}$, i.e. ${\mk s}_{\ket{\Psi}}$ (see Lemma \ref{thm:image}). The relative simplicity of this case is due to the fact that $\mathfrak{s}_{\mu(\ket{\Psi})}$ is diagonal and hence $\mathfrak{s}_{\ket\Psi}$ is diagonal too. For degenerate NONs Lemma \ref{thm:image} is just the first step and the procedure is much more complicated. In order to deal with degenerate NONs we first introduce the notion of adapted NOs (see Definition 5). This allows us to formulate Lemma 7 which charaterizes $\im\ d\mu_{\ket{\Psi}}\cap \mk{d}$. In contrast to nondegenerate case we cannot, however, use this lemma to immediately formulate a selection rule. In order to deduce the support of $\kpsi$ from lemma \ref{lemma:image-general} and the knowledge of its NONs, we have to take a closer look at the subtle structure of local symmetries of states with fixed NONs. The main tool we use is the idea that fibres of the momentum map are stratified symplectic spaces. This lead us to additional technical Assumption \ref{assumption} which, as we prove in Section \ref{subsec:examples} is always satisfied for many qubit systems, as well as, for all fermionic systems with $N\leq 5$ and $d\leq 11$. Under this additional assumption we formulate a selection rules for degenerate NONs in theorem \ref{thm:maximal-ansatz}.  Finally, we conjecture that Assumption \ref{assumption} is satisfied for any multipartite system. Proving this is a challenging combinatorial problem which we leave open.

It is worth to mention that the approach taken in this paper cannot be easily extended to the overlapping quantum marginal problem. This is due to the fact that the corresponding symmetry groups that give rise to the momentum map and thus to the reduced density matrices do not commute.

\section*{Acknowledgments}
We acknowledge financial support from National Science Centre, Poland under the grant SONATA BIS: 2015/18/E/ST1/00200 (AS),
the Excellence Initiative of the German Federal and State Governments (Grants ZUK 43 \& 81) (DG \& AL), and the DFG (project B01 of CRC 183) (DG), the UK Engineering and Physical Sciences Research Council (Grant EP/P007155/1) and the German Research Foundation (Grant SCHI 1476/1-1) (CS).

\appendix
\section{Checking the validity of Assumption \ref{assumption}}
In this section we shall describe an efficient combinatorial procedure for checking the validity of assumption \ref{assumption}. The naive approach would amount to realising a loop going through all subsets of vertices of the Pauli hypercube that span a hyperplane and checking whether the convex hull of all vertices contained in such a hyperplane contains vector $\bd n$. However, in practice, it turns out that a much stronger condition is satisfied. We will explain it by exploring the geometry of the spectral polytope around point $\bd n$. To this end, for each saturated GPC and each generator of group $\Pi_{\bd{n}}$ we define the following region (see Fig.\ref{fig:appendix})
\begin{eqnarray}
\fl\nonumber R^{(k,j)}_{\bd n}:=\{\bd{n}':\ D_k(\bd{n}')>0\mathrm{\ and}\ D_k\circ \pi_{j,j+1}(\bd{n}')<0\}\cup\\\cup \{\bd{n}':\ D_k(\bd{n}')<0\mathrm{\ and}\ D_k\circ \pi_{j,j+1}(\bd{n}')>0\}.
\end{eqnarray}
\begin{figure}[h]
\centering
\includegraphics[width=0.6\textwidth]{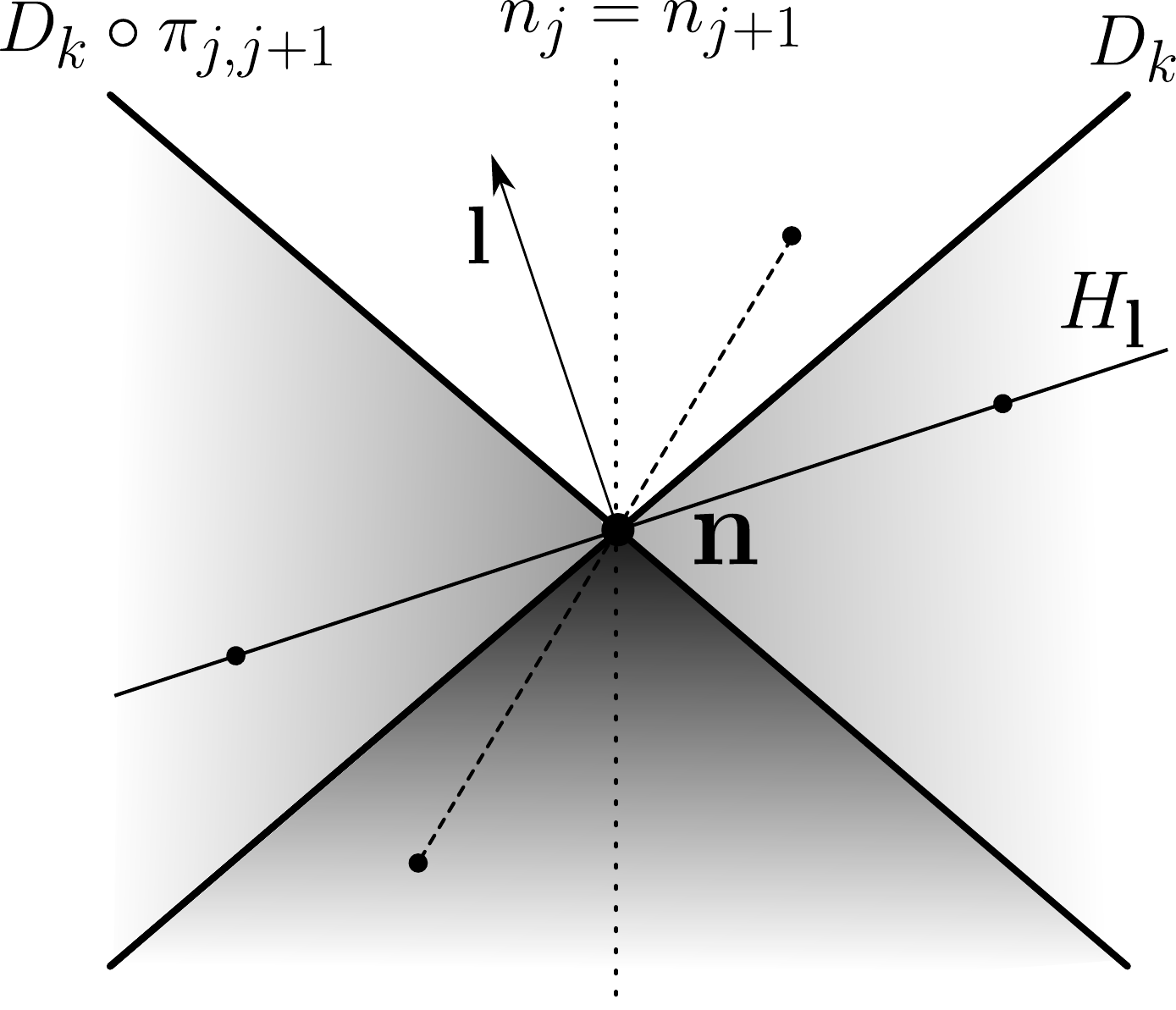}
\caption{Region $R^{(k,j)}_{\bd n}$ is marked with light grey. Dark grey area corresponds to the region where the presence of vertices of the Pauli hypercube is allowed. This is because they span hyperplanes that are not fully contained in the region of admissible one-particle spectra around point $\bd n$ (dashed line segment).}
\label{fig:appendix}
\end{figure}

Hyperplane $H_{\bd{l}}:=\{\bd{n}':\ \bd{l}\cdotp (\bd{n}'-\bd{n})=0\}$ (see Fig.\ref{fig:appendix}) that would lead to an ansatz which is different than ansatzes $\{A_k\}$ coming from the saturated GPCs, necessarily has to be contained in region
\begin{eqnarray}
R_{\bd n}:=\bigcup_{k:\ D_k(\bd n)=0}\bigcup_{j:\ \pi_{j,j+1}\in \Pi_{\bd{n}}}R^{(k,j)}_{\bd n}
\end{eqnarray}
To check validity of assumption \ref{assumption}, for each $R^{(k,j)}_{\bd n}$ we check numerically that point $\bd{n}$ is not contained in the convex hull of vertices of the Pauli hypercube that are contained in chosen $R^{(k,j)}_{\bd n}$. This implies that it is not possible to find a hyperplane which is contained in region $R_{\bd n}$ and is spanned by vertices of the Pauli hypercube.
\vspace{2.0cm}

\bibliographystyle{unsrt}
\bibliography{pin}

\begin{thebibliography}{10}

\bibitem{Loewdin55}
P.-O. L\"owdin.
\newblock Quantum theory of many-particle systems. {I}. physical
  interpretations by means of density matrices, natural spin-orbitals, and
  convergence problems in the method of configurational interaction.
\newblock {\em Phys. Rev.}, 97:1474, 1955.

\bibitem{DavidsonRev}
E.~R. Davidson.
\newblock Properties and uses of natural orbitals.
\newblock {\em Rev. Mod. Phys.}, 44:451, 1972.

\bibitem{Pin1}
C.~Schilling, C.~L. Benavides-Riveros, A.~Lopes, T.~Maci\k{a}\.zek, and
  A.~Sawicki.
\newblock Implications of pinned occupation numbers for natural orbital
  expansions. i: Generalizing the concept of active spaces.
\newblock {\em arXiv:1908.10938}, 2019.

\bibitem{Guillemin}
V.~Guillemin and S.~Sternberg.
\newblock Convexity properties of the moment mapping.
\newblock {\em Invent. Math.}, 67:491--514, 1982.

\bibitem{SWK13}
A.~Sawicki, M.~Walter, and M~Ku\'{s}.
\newblock {When is a pure state of three qubits determined by its
  single-particle reduced density matrices?}
\newblock {\em J. Phys. A: Math. Theor.}, 46:055304, 2013.

\bibitem{MOS13}
T.~Maci\k{a}\.{z}ek, M.~Oszmaniec, and A.~Sawicki.
\newblock How many invariant polynomials are needed to decide local unitary
  equivalence of qubit states?
\newblock {\em J. Math. Phys.}, 54(9):092201, 2013.

\bibitem{Sjamaar}
R.~Sjamaar and E.~Lerman.
\newblock Stratified symplectic spaces and reduction.
\newblock {\em Ann. of Math.}, 134 (2):375--422, 1991.

\bibitem{S18}
A.~Sawicki, T.~Maci\k{a}\.{z}ek, M.~Oszmaniec, K.~Karnas, K.~Kowalczyk-Murynka,
  and M.~Ku\'s.
\newblock Multipartite quantum correlations: symplectic and algebraic geometry
  approach.
\newblock {\em Reports on Mathematical Physics}, 82(1):81--111, 2018.

\bibitem{MS18}
T.~Maci\k{a}\.zek and A.~Sawicki.
\newblock {Asymptotic properties of entanglement polytopes for large number of
  qubits.}
\newblock {\em J. Phys. A: Math. Theor}, 7(51):07LT01, 2018.

\bibitem{TomekDoub}
T.~Maci\k{a}\.{z}ek and V.~Tsanov.
\newblock Quantum marginals from pure doubly excited states.
\newblock {\em J. Phys. A}, 50:465304, 2017.

\bibitem{vergne2017inequalities}
M.~Vergne and M.~Walter.
\newblock Inequalities for moment cones of finite-dimensional representations.
\newblock {\em Journal of Symplectic Geometry}, 15(4):1209--1250, 2017.

\bibitem{MS15}
T.~Maci\k{a}\.{z}ek and A.~Sawicki.
\newblock Critical points of the linear entropy for pure l-qubit states.
\newblock {\em J. Phys. A}, 48(4):045305, 2015.

\bibitem{OSS14}
M.~Oszmaniec, A.~Sawicki, and P.~Suwara.
\newblock Geometry and topology of cc and cq states.
\newblock {\em J. Math. Phys.}, 55(03):062204, 2014.

\bibitem{SOK14}
A.~Sawicki, M.~Oszmaniec, and M.~Ku\'s.
\newblock Convexity of momentum map, morse index, and quantum entanglement.
\newblock {\em Rev. Math. Phys.}, 26(03):1450004, 2014.

\bibitem{ST13}
A.~Sawicki and V.~Tsanov.
\newblock A link between quantum entanglement, secant varieties and sphericity.
\newblock {\em J. Phys. A}, 46(26), 2013.

\bibitem{HKS13}
A.~Huckleberry, M.~Ku\'s, and A.~Sawicki.
\newblock Bipartite entanglement, spherical actions, and geometry of local
  unitary orbits.
\newblock {\em J. Math. Phys.}, 54(2):022202, 2013.

\bibitem{walter2013entanglement}
M.~Walter, B.~Doran, D.~Gross, and M.~Christandl.
\newblock Entanglement polytopes: multiparticle entanglement from
  single-particle information.
\newblock {\em Science}, 340(6137):1205--1208, 2013.

\bibitem{SOK12}
A.~Sawicki, M.~Oszmaniec, and M.~Ku\'s.
\newblock Critical sets of the total variance of state detect all slocc
  entanglement classes.
\newblock {\em Phys. Rev. A}, 86:040304(R), 2012.

\bibitem{SK11}
A.~Sawicki and M.~Ku\'s.
\newblock Geometry of the local equivalence of states.
\newblock {\em J. Phys. A}, 44(49), 2011.

\bibitem{SHK11}
A.~Sawicki, A.~Huckleberry, and M.~Ku\'{s}.
\newblock Symplectic geometry of entanglement.
\newblock {\em Comm. Math. Phys.}, 305:441--468, 2011.

\bibitem{Higushi}
A.~Higuchi, A.~Sudbery, and J.~Szulc.
\newblock One-qubit reduced states of a pure many-qubit state: Polygon
  inequalities.
\newblock {\em Phys. Rev. Lett.}, 90:107902, Mar 2003.

\bibitem{FTthesis}
F.~Tennie.
\newblock {\em Influence of the exchange symmetry beyond the exclusion
  principle}.
\newblock PhD thesis, University of Oxford, 2017.

\bibitem{Coleman}
A.~J. Coleman and V.~I. Yukalov.
\newblock {\em Reduced Density Matrices: Coulson's Challenge}.
\newblock Springer, New York, 2000.

\bibitem{Mazz04}
D.~A. Mazziotti.
\newblock Realization of quantum chemistry without wave functions through
  first-order semidefinite programming.
\newblock {\em Phys. Rev. Lett.}, 93:213001, Nov 2004.

\bibitem{Mazz12}
D.~A. Mazziotti.
\newblock Structure of fermionic density matrices: Complete
  $n$-representability conditions.
\newblock {\em Phys. Rev. Lett.}, 108:263002, Jun 2012.

\bibitem{Mazz16}
D.A. Mazziotti.
\newblock Pure-$n$-representability conditions of two-fermion reduced density
  matrices.
\newblock {\em Phys. Rev. A}, 94:032516, Sep 2016.

\bibitem{CSHCB}
F.~Tennie, V.~Vedral, and C.~Schilling.
\newblock Universal upper bounds on the {B}ose-{E}instein condensate and the
  {H}ubbard star.
\newblock {\em Phys. Rev. B}, 96:064502, 2017.

\end{thebibliography}

\end{document}